\documentclass[twoside,leqno]{article}

\usepackage[letterpaper]{geometry}

\usepackage{siamproceedings}

\usepackage[T1]{fontenc}
\usepackage{graphicx}
\usepackage{epstopdf}
\usepackage{enumitem}
\usepackage{algorithmic}
\ifpdf
  \DeclareGraphicsExtensions{.eps,.pdf,.png,.jpg}
\else
  \DeclareGraphicsExtensions{.eps}
\fi

\usepackage{amsmath,amssymb,euscript,enumitem,amsfonts}
\usepackage{tcolorbox}
\usepackage[mathlines]{lineno}
\linenumbers
\usepackage{thmtools} 
\usepackage{thm-restate}

\usepackage{abstract}
\renewcommand{\abstractname}{} 

\usepackage{graphicx}
\usepackage{tikz}
\usetikzlibrary{arrows.meta}
\usetikzlibrary{calc}
\usetikzlibrary{decorations.pathmorphing}
\usetikzlibrary{positioning}
\usepackage{pgfplots}
\pgfplotsset{compat=1.18}
\usepackage{subcaption}

\usepackage{hyperref}
\usepackage{xcolor}
\definecolor{darkgreen}{RGB}{0,100,0}
\definecolor{darkred}{RGB}{139,0,0}

\newsiamremark{remark}{Remark}
\newsiamremark{observation}{Observation}
\newsiamthm{claim}{Claim}

\newcommand{\M}{\mathcal{M}}

\newcommand{\symm}{\operatorname{sym}}
\newcommand{\Gsymm}{\operatorname{gsym}}
\newcommand{\foc}{\operatorname{foc}}

\newcommand{\Alpha}{A}
\newcommand{\Beta}{B}

\newcommand{\Skip}[1]      {}

\def\barroman#1{\sbox0{#1}\dimen0=\dimexpr\wd0+1pt\relax
  \makebox[\dimen0]{\rlap{\vrule width\dimen0 height 0.06ex depth 0.06ex}%
    \rlap{\vrule width\dimen0 height\dimexpr\ht0+0.03ex\relax 
            depth\dimexpr-\ht0+0.09ex\relax}%
    \kern.5pt#1\kern.5pt}}

\title{The Singular Source of Vineyard Monodromy}

\author{Erin W.~Chambers\thanks{Department of Computer Science and Engineering, University of Notre Dame, Notre Dame, IN, USA (\email{echambe2@nd.edu},  \url{https://wolfchambers.github.io/}, \url{https://orcid.org/0000-0001-8333-3676}).}
\and Christopher Fillmore\thanks{Institute of Science and Technology Austria, Klosterneuburg, Austria (\email{cdfillmore@gmail.com},  \url{https://orcid.org/0000-0001-7631-2885}).}   
\and Shankha Shubhra Mukherjee \thanks{ Department of Computer Science and Engineering, University of Notre Dame, Notre Dame, IN, USA (\email{smukher4@nd.edu},  \url{https://shankhasm.github.io}, \url{https://orcid.org/0009-0009-6779-059X}).}
\and Rohit Roy \thanks{Inria Centre Universit{\'e} C{\^o}te d'Azur, Sophia-Antipolis, France
(\email{rohit.a.roy@inria.fr}, \url{https://rohitroy.me}, \url{https://orcid.org/0009-0002-1927-2172}).}
\and Elizabeth Stephenson\thanks{Orteliu, Oslo, Norway
(\email{elizasteprene@gmail.com}, \url{https://orcid.org/0000-0002-6862-208X}).}
\and Mathijs Wintraecken\thanks{Inria Centre Universit{\'e} C{\^o}te d'Azur, Sophia-Antipolis, France
(\email{mathijs.wintraecken@inria.fr}, \url{https://orcid.org/0000-0002-7472-2220}).}}

\begin{document}
\maketitle
\begin{abstract}


Vineyards, or time-varying families of persistence diagrams, are widely used in topological data analysis (TDA) pipelines to track how topological features change and evolve as a parameter varies.  When the parameter traces a closed loop, a vineyard can exhibit monodromy: diagram points permute over the course of a full traversal, which obstructs feature tracking and can complicate downstream analysis of such data.
Chambers et al. considered the periodic vineyards that arise from the radial persistence transform, which maps the manifold to a family of persistence diagrams, where each diagram fixes a base point and considers the filtration that is based on Euclidean distance to that point, and showed that monodromy and knotting can occur. Other recent work by Arya et al. considers geometric conditions that exclude monodromy in two dimensions, in an effort to better understand when this effect happens.
That said, understanding when and why monodromy occurs is a fundamental open problem with direct practical consequences for many data analysis pipelines.  
In this work, we study this question for 1-manifolds in $\mathbb{R}^2$, using a surprising connection with tools from singularity theory, and provide a classification for the causes of monodromy in vineyards.  More precisely, we prove that the vineyard of a sufficiently small loop $\gamma$ cannot exhibit monodromy unless it contains a specific singularity of the distance function.  The central geometric object in our analysis is the symmetry set, which is the locus of centers of spheres tangent in more than one point to the manifold; this object classifies singularities of the distance function, and in our setting, dictates precisely when monodromy occurs.  This characterization opens the door to the development of algorithmic criteria for detecting and utilizing (or avoiding) monodromy in TDA pipelines.



\end{abstract}

\section{Introduction}
\label{sec:Introduction}

Topological Data Analysis (TDA) encompasses a wide range of tools that compute topological invariants of geometric objects or spaces using tools from algebraic topology. Persistent homology in particular is perhaps the most well-known example in TDA, yielding a computable and stable invariant known as the persistence diagram, which has seen wide utility in data analysis in a broad range of applications.  A full review of such applications is beyond our scope, but we refer the reader to~\cite{DONUT} for a list of further examples; see also recent books on this topic~\cite{Boissonnat2018,Dey2022}.

Time-series data yields a family of stacked persistence diagrams known as vineyards~\cite{CohenSteiner2006, turner2023representing}, which have been useful in many applications \cite{Bergomi2020, DeeAlgar2021, Yoo2016}. 
Because of the stability of persistence diagrams~\cite{CohenSteinerStability}, the points in the stacked persistence diagrams move continuously over time. This means that we can follow a point in (the stack of) the persistence diagrams; the resulting curve is called a vine.  Vineyards can be computed via a modified version of the standard persistence algorithm~\cite{CohenSteiner2006}.
Here we will call the family or `circular stack' of persistence diagrams of a continuous periodic family of filtrations a closed vineyard~\cite{Burghelea2013}. 

\subsubsection*{Monodromy and topology in vineyards:} Recently, theoretical interest in vineyards has expanded in several exciting ways. In \cite{Arya2024}, an example of a closed vineyard with monodromy was exhibited, and the authors proved that the persistent
homology transform of a star-shaped object in the plane cannot exhibit monodromy, giving the first hints of connecting the geometry of a space to its vineyard. 
In the future work section, they note 
a number of potential generalizations of this result, especially to higher dimensions, but were unable to characterize obstructions to monodromy using their approach for two dimensions. 
We note that monodromy in multipersistence was already shown in \cite{Cerri2013}, see also \cite{AATRNsara}, but not much was understood about its causes or obstructions.

Following this work, in~\cite{Chambers2026} circular vineyards were shown to be as topologically complicated as one could hope or expect. In particular, given a knot (or even a link), a manifold $\M$ and a periodic family of distance functions restricted to the manifold can be constructed such that the closed vineyard contains the given knot (or link). 
The distance function in the construction is the squared Euclidean distance to the point $\gamma(t)$ on a closed curve $\gamma$, $d_\mathbb{E}^2  (\cdot , \gamma(t)) |_\M =(d_\mathbb{E}  (\cdot , \gamma(t)) )^2 |_\M$. We denote the resulting vineyard based on this loop by $\mathcal{V}(\mathcal{M},\gamma)$. 
Because recognizing knots is related to many NP-hard problems \cite{Ichihara2023,hass1997algorithms, Lackenby2021, Koenig2021}, this means that closed vineyards are also computationally complex, in the sense that comparing them is hard.  This potential obstruction is perhaps especially surprising given the practical utility of such vineyards in a number of applied settings.

\subsubsection*{The persistent homology transform:} 
Persistence-based transforms are a relatively new development in TDA, but are of growing interest based on their strength both in theory and practice.  For example, it is now well known that the directional transform, where each direction $\omega$ gives rise to a sublevel set filtration and hence a persistence diagram, is injective~\cite{Curry2022,Ghrist2018}. In order to completely determine a piecewise linear shape, an exponential number of directions is sufficient for most settings~\cite{Curry2022}, but no precise lower bound is known outside of some special settings~\cite{Belton2020,Hofer2019,Fasy2025,Chambers2026-counting}.  One interesting expansion of the original directional transform uses extended persistence rather than traditional persistence, as we will utilize in this paper; see Section \ref{ssec:persistence} for a more detailed definition and comparison.  This variant has shown considerable promise and even provides better results for some applications~\cite{Turner2024}.
Later generalizations of directional transforms include the radial transform, which we will utilize, as well as the distance-to-a-flat filtration~\cite{Shoving}, which are also injective and in some sense practically more powerful and often faster to compute than the directional transform in terms of the number of filtrations needed. 

However, a completely geometric characterization of the set of critical directions in any setting is unknown. While it is well understood that monodromy in general is connected to these singularities, which are known as the discriminant locus in the algebraic setting~\cite{salter2023stratifiedbraidgroupsmonodromy,Salter2024}, to the best of our knowledge, there is no work exploring these structures in the context of persistence.  In the context of catastrophe theory, this is called the bifurcation set, see e.g. \cite{Arnold2012Ch3, Bruce1985symmetry}. Notably, this is the nomenclature that Bruce, Giblin and Gibson used \cite{Bruce1985symmetry}, and that we will also adopt for the remainder of the paper.
Our work in this paper is motivated by the fact that there is not yet a concrete understanding of the singularities of the transforms or how the geometry of a shape gives rise to monodromy in the vineyard of a transform.

\subsubsection*{Our contribution:} 
This paper explores the local geometric origins of monodromy in closed vineyards from a different perspective, given a more complete understanding of the bifurcation set and the implications of singularity theory when applied to our setting. More precisely, we will exhibit that only very specific singularities of the symmetry set (and focal set) generate non-trivial topological changes in the vineyards. 

\begin{figure}[th!] 
    \centering
    \includegraphics[width=0.9\linewidth]{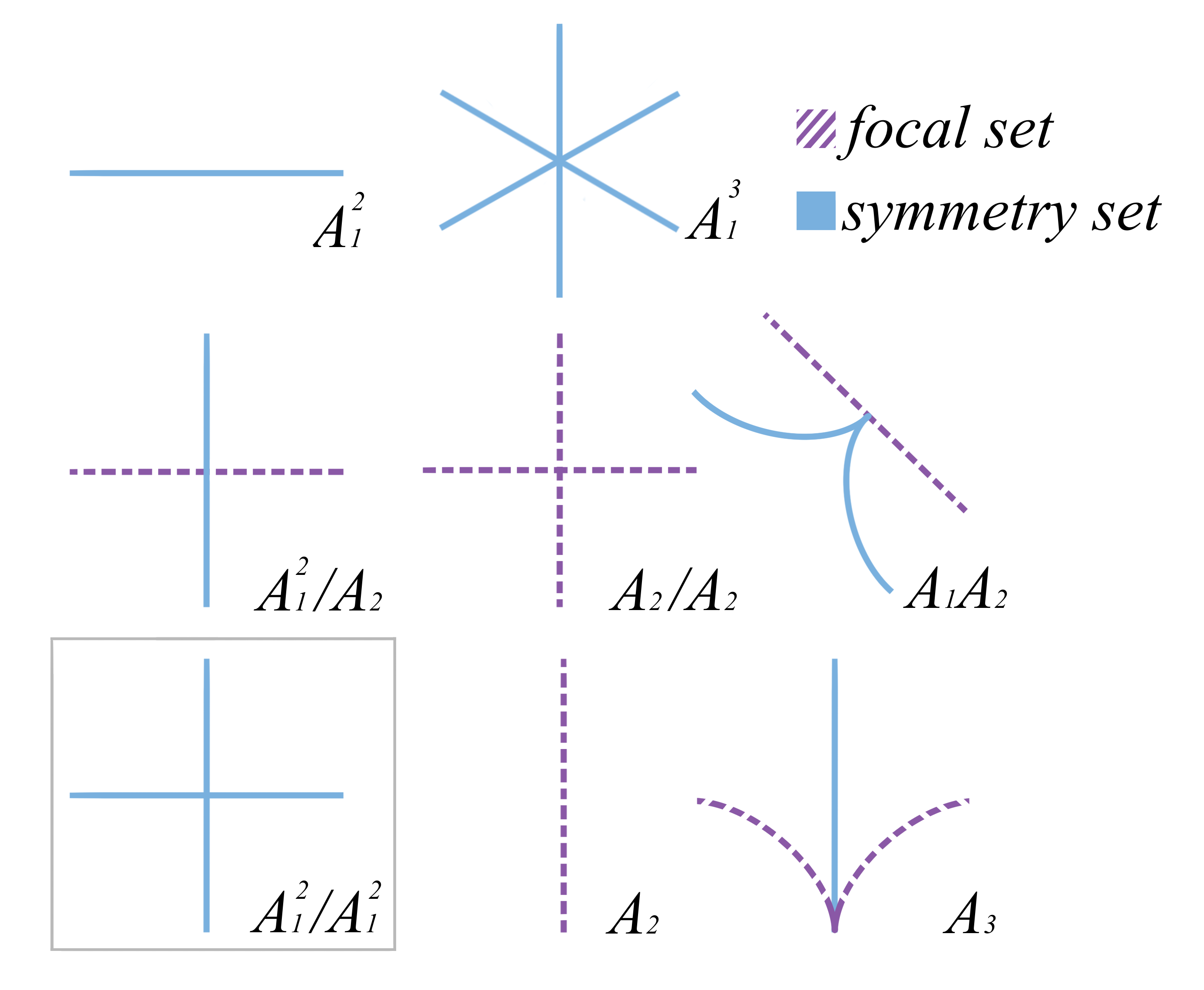}
    \caption{The singularities of the symmetry and focal set in the plane. We adopt Arnold's notation for these singularities. The lower left boxed singularity is the only one which can exhibit local monodromy. 
    \label{fig:classification2DEarly} } 
\end{figure}

We recall that given a set $\mathcal{S} \subset \mathbb{R}^d$, the medial axis is defined as the set of points in $\mathbb{R}^d$, such that the ball centred at those points intersects $\mathcal{S}$ in multiple points and the intersection of the interior with $\mathcal{S}$ is empty. The symmetry set \cite{Bruce1985symmetry} of a manifold $\mathcal{M}$ is the set of centres of spheres that are tangent to $\mathcal{M}$ in multiple points; hence, the symmetry set contains the medial axis. Finally, the focal set is the (image of the) set of points where the natural map from the normal bundle of a submanifold to Euclidean space is not an immersion.  If we take the union of these objects, we obtain the generalized symmetry set, which classifies all singularities of the distance function to the manifold in the ambient space.
The singularities in our setting will consist of the points where the generalized symmetry set is not manifold. These generic singularities have been classified in the plane \cite{Bruce1985symmetry} and the different classes are denoted using Arnold's notation ($A_1^2$, $A_1^2/A_1^2$, $A_1^3$, $A_2$, $A_1A_2$, $A_2/A_2$, $A_1^2/A_2$, $A_3$). The classification is depicted in Figure \ref{fig:classification2DEarly};
see also Section \ref{sec:Singularities_Plane} for a description of the classification in detail.

With this notation, we can state the main result of this paper: 
\begin{restatable}[$A_1^2/A_1^2$ is the unique local planar monodromy generator]{theorem}{mainthm}
\label{thm:mainthm}
Let $\mathcal{M}\subset\mathbb{R}^2$ be a generic smooth closed curve and let $\gamma:S^1\to\mathbb{R}^2$ be a generic sufficiently small loop.
If the interior of $\gamma$ contains no singularity of type $A_1^2/A_1^2$ of the symmetry set, then the vineyard
$\mathcal{V}(\mathcal{M},\gamma)$ has no nontrivial monodromy, and the vineyard consists of $k$ unlinked circles. 
\end{restatable}

\begin{remark} 
We note that sufficiently small in this context means that the symmetry and/or focal set in the interior of the loop is diffeomorphic to one of the model singularities given in the classification \cite{Bruce1985symmetry} as depicted in Figure \ref{fig:classification2DEarly}. 
\end{remark}

\begin{remark} 
We also note that in contrast to prior work~\cite{Chambers2026}, in our setting vines can start or end on the diagonal, in which case we extend them on the diagonal to form a loop; see Section~\ref{sec:MonodromyAndVineyards} for a precise definition and full details. 
\end{remark}

While our analysis is restricted to $\mathbb{R}^2$, we note that 2D already captures the core phenomenon.  In particular, all singularities essential for monodromy along a loop are already present in this dimension~\cite{Arnold1972,Brocker1975,Bruce1985symmetry}. 

Although we know of one work that studies the medial axis and symmetry set from the perspective of persistent homology \cite{edelsbrunner2025mid}, we are not aware of any prior work which combines the symmetry set, persistence, and singularity theory.  Of course, the symmetry set and medial axis have been studied from a singularity theory perspective, see e.g. \cite{damon2017medial, damon2006global, yomdin1981local, mather1983distance} and the symmetry set, in fact, has its roots in this community \cite{Looijenga1974Thesis, Bruce1985symmetry}.
The medial axis has also seen many successful applications in computational geometry and topology more generally, and serves as the basis for algorithms in topological inference as well as shape reconstruction and  simplification~\cite{Amenta1999,Niyogi2008,Chazal2006}. 
Beyond the direct result, namely the identification of the source of non-trivial topology and monodromy in vineyards, we believe that this connection is the most consequential result of this paper, which may open up many new and exciting mathematical directions to consider.

In addition, our classification potentially provides a framework that can be utilized to design algorithms that ensure the presence of monodromy, opening the door to tools from knot theory to find applications in the analysis of vineyards.  At the same time, it could also provide insights into how to construct vineyards which avoid knotting and monodromy, in which case feature tracking is considerably simplified.  Both avenues of research could potentially have interesting applications in the analysis of real-world data.


\section{Preliminaries}
\label{sec:preliminaries}

\subsection{Focal and Symmetry Set}

In 2D, the focal set is also called the \textit{evolute} of the curve. The evolute is defined as the locus of the centers of curvature of the curve, see e.g. \cite{porteous2001geometric}. 

Let $\mathcal{M}$ be a smooth plane curve parameterized by arc length $s$, given by $\gamma(s)$. Let $\mathbf{t}(s) = \gamma'(s)$ be the unit tangent vector and $\mathbf{n}(s)$ be the unit normal vector obtained by rotating $\mathbf{t}(s)$ by $90^\circ$ counter-clockwise.

\begin{definition}[Curvature and Radius of Curvature]
The \textbf{curvature} $\kappa(s)$ is the rate of change of the tangent vector in the direction of the normal:
\begin{equation*}
    \mathbf{t}'(s) = \kappa(s) \mathbf{n}(s)
\end{equation*}
The \textbf{radius of curvature} $\rho(s)$ is the reciprocal of the curvature:
\begin{equation*}
    \rho(s) = \frac{1}{\kappa(s)}, \quad \text{for } \kappa(s) \neq 0
\end{equation*}
\end{definition}

\begin{definition}[Evolute or focal set]
The \textbf{evolute}, or \textbf{focal set} $\foc(\M)$, of $\mathcal{M}$ is the curve $E(s)$ defined by the centers of osculating circles {of the plane curve} $\mathcal{M}$:
\begin{equation*}
    E(s) = \gamma(s) + \rho(s) \mathbf{n}(s)
\end{equation*}
where $\mathbf{n}(s)$ is the unit normal vector to the curve at $\gamma(s)$.
\end{definition}

Geometrically, the evolute can also be characterized as the envelope of the family of lines normal to $\mathcal{M}$. Additionally, it can be viewed as the locus of the centres of circles whose 2-jet coincides with the curve at the point of contact of the curve and the circle.  We note that term evolute is used exclusively in 2D, while the focal set generalizes this concept to arbitrary dimensions.

\begin{definition}[Symmetry set, \cite{Bruce1985symmetry}] \label{def:symmetry_set}
Let $\M \subset \mathbb{R}^d$ be a smooth manifold. The symmetry set $\symm (\M)$ of $\M$ is the closure of the set of points $p$ in $\mathbb{R}^d$, such that there is a radius $r$ so that $B(p,r)$ is tangent to $\M$ in at least two places.
\end{definition}

\begin{figure}[htbp]
    \centering
    \includegraphics[width=0.54\textwidth]{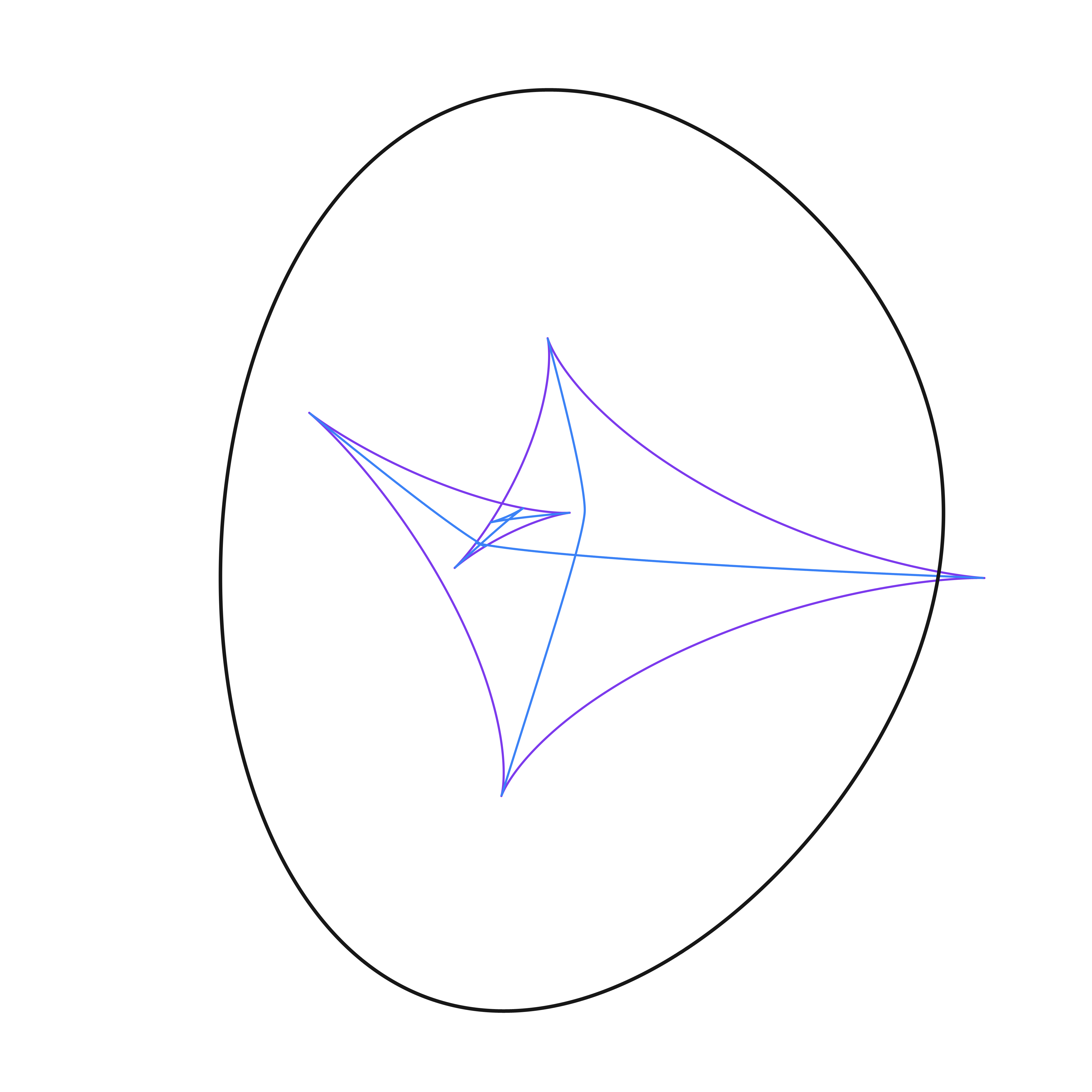}
        \caption{Visualizing the focal set (evolute) and symmetry set of a spline-defined curve $\mathcal{M}$. The purple locus represents the evolute. The blue locus represents the symmetry set.}
    \label{fig:focal_set_visualizer}
\end{figure}

\begin{definition}[Generalized symmetry set]
Let $\M \subset \mathbb{R}^d$ be a smooth manifold. The generalized symmetry set $\Gsymm (\M)$ of $\M$ is the union of $\symm (\M)$ and  $\foc(\M)$. 
\end{definition}

\subsection{Persistence and extended persistence}
\label{ssec:persistence}


\begin{figure}[tb]
    \centering

    \begin{subfigure}[t]{0.50\textwidth}
        \centering
        \includegraphics[width=\linewidth]{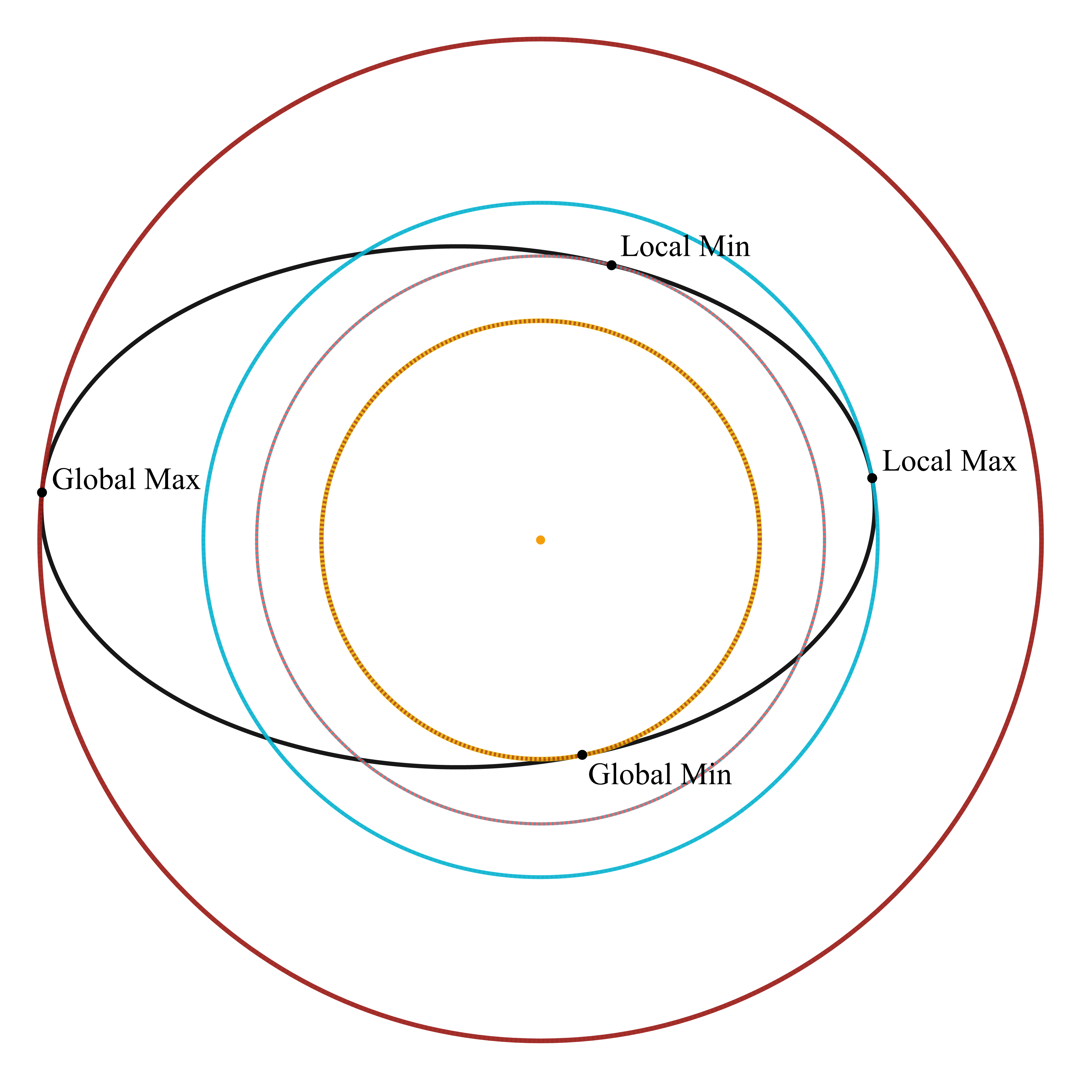}
        \caption{Radial filtration of the ellipse.}
        \label{fig:first-image}
    \end{subfigure}
    \hfill
    \begin{subfigure}[t]{0.48\textwidth}
        \centering
        \includegraphics[width=\linewidth]{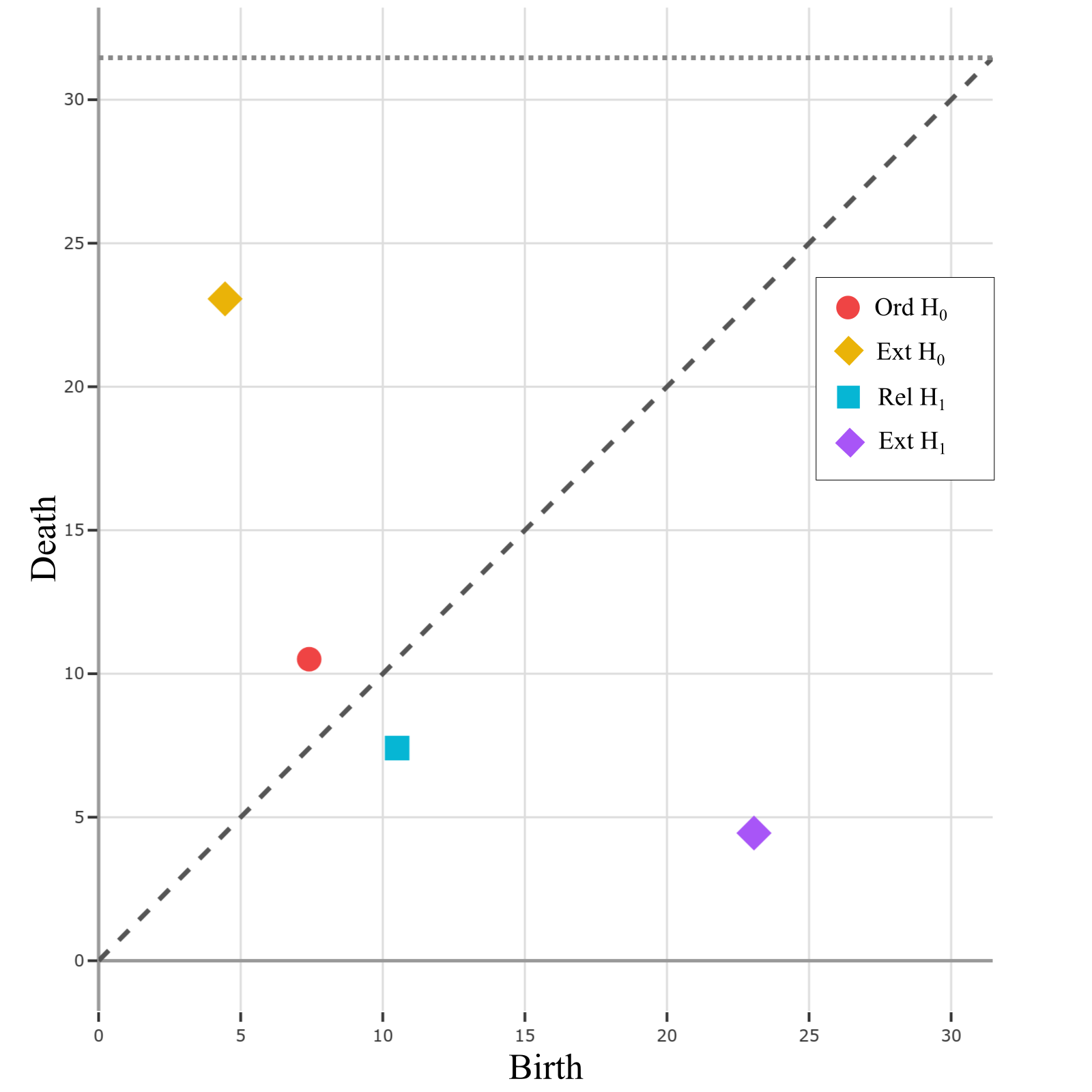}
        \caption{Extended persistence diagram of the ellipse.}
        \label{fig:second-image}
    \end{subfigure}

    \caption{Radial persistence of an ellipse (shown in black) with respect to a base point (shown in yellow). Note that each of the four critical radii  contribute to two persistence points (one in the upward filtration and one on the downward relative filtration), for a total of four points, two of which are above the diagonal and two below. }
    \label{fig:2dpersistence}
\end{figure}

Given a manifold $\mathcal{M}$ and a function $f \colon \mathcal{M} \to \mathbb{R}$, persistent homology tracks how the topology of the sublevel sets $\mathcal{M}_a = f^{-1}(-\infty, a]$ evolves as $a$ increases~\cite{Edelsbrunner2002,Robins1999}; this results in a filtration where each space $M_a$ includes into later spaces $M_{a'}$ where $a<a'$. When $f$ is a Morse function, each critical point either creates a new homology class or destroys an existing one. We say that a $\ell$-dimensional class $\alpha \in H_{\ell}(\mathcal{M}_{a_i})$ is \emph{born} at $a_i$ if it does not appear in the image of the inclusion-induced map from $H_{\ell}(\mathcal{M}_{a_{i-1}})$, and \emph{dies} at $a_j$ if it first becomes trivial relative to an older class at that level. Pairing each birth with its corresponding death yields the \emph{persistence diagram} $\mathrm{Dgm}_{\ell}(f)$, a multiset of points $(b, d) \in \mathbb{R}^2$ above the diagonal; the value $d - b$ measures the lifetime, or \emph{persistence}, of the corresponding topological feature.  

A limitation of standard persistence in our context is that some homology classes are born but never die within the sublevel set filtration.  For instance, a nontrivial loop in a surface is present across all larger sublevel sets after its birth, and hence is said to contribute a point at infinity. Extended persistence resolves this by appending a descending filtration through relative homology groups, applying Lefschetz and Poincar\'e duality~\cite{CohenSteiner2008,Agarwal2006}, so that after the sublevel set filtration reaches $\mathcal{M}$, we append the groups $H_p(\mathcal{M}, \mathcal{M}^a)$ to the filtration, where $\mathcal{M}^a = f^{-1}[a, \infty)$. The resulting extended filtration sequence
\begin{equation*}
    0 \to H_p(\mathcal{M}_{a_1}) \to \cdots \to H_p(\mathcal{M}) \to
    H_p(\mathcal{M}, \mathcal{M}^{a_k}) \to \cdots \to H_p(\mathcal{M}, \mathcal{M}) = 0
\end{equation*}
begins and ends at zero, guaranteeing that every homology class born in the filtration eventually dies at a finite value.  See Figure~\ref{fig:2dpersistence}.

The pairs arising from this construction fall into three types: \emph{ordinary} pairs, born and dying in the upward (sublevel set) sweep; \emph{relative} pairs, born and dying in the downward (relative) sweep; and \emph{extended} pairs, that are born on the upward sweep and die on the downward. Ordinary pairs lie above the diagonal; relative pairs below; and extended pairs may appear on either side.  Note that for manifolds without boundary, the persistent points are mirrored across the diagonal; see again Figure~\ref{fig:2dpersistence}.  Together, these points provide a complete pairing of the critical points of~$f$.

\subsection{Monodromy and vineyards}
\label{sec:MonodromyAndVineyards}


Let $\M \subset \mathbb{R}^2$ be a smooth closed curve and let
\[
\gamma \colon [0,2\pi] \to \mathbb{R}^2
\]
be a smooth loop, so that $\gamma(0)=\gamma(2\pi)$. For each $t \in [0,2\pi]$ we consider the
function\footnote{We work with the squared distance, instead of the distance (which is used in e.g. \cite{Chambers2026}),  because the singularity theoretic classification recalled in
Section~\ref{sec:Singularities_Plane} is naturally phrased for $d_{\mathbb{E}}(\,\cdot\,,p)^2\big|_\M$. Since the map
$r \mapsto r^2$ is strictly increasing on $[0,\infty)$, the sublevel set filtrations of
$d_{\mathbb{E}}(\,\cdot\,,\gamma(t))\big|_\M$ and $d_{\mathbb{E}}(\,\cdot\,,\gamma(t))^2\big|_\M$
differ only by a reparametrization of the filtration parameter. Hence the combinatorics of births,
deaths, vines, and monodromy are the same in either convention.}
\[
f_t \colon \M \to \mathbb{R}, \qquad
f_t(x) = d_{\mathbb{E}}(x,\gamma(t))^2\big|_\M .
\]

Fix a homological degree $\ell \geq 0$. For each $t$, let
\[
\mathrm{Dgm}_{\ell}(f_t)
\]
denote the (extended) $\ell$-dimensional persistence diagram of the sublevel set filtration induced by $f_t$.
When the homological degree is clear from the context, we suppress the subscript $\ell$ and write
simply $\mathrm{Dgm}(f_t)$. We also write
\[
\Delta = \{(b,d)\in \mathbb{R}^2 \mid b=d\}
\]
for the diagonal.

\textbf{Throughout this paper we assume that the family is generic in the following sense:} each
persistence diagram $\mathrm{Dgm}_{\ell}(f_t)$ contains only finitely many points and has no point of
multiplicity larger than one. In particular, away from the diagonal, the points in the persistence
diagram move continuously with $t$ by stability of persistence diagrams~\cite{CohenSteinerStability},
and the resulting vineyard has no self-intersections away from $\Delta$. This is the generic
situation for vineyards, as noted in prior work~\cite{turner2023representing}.

\begin{definition}[Closed vineyard and vine]
For a fixed degree $\ell$, the \emph{closed $\ell$-vineyard} associated to $(\M,\gamma)$ is
\[
\mathcal{V}_{\ell}(\M,\gamma)
=
\left\{ (t,z)\in S^1\times \mathbb{R}^2 \; { \big |} \;
z\in \mathrm{Dgm}_{\ell}(f_t) \right \},
\]
where we identify $S^1$ with $[0,2\pi]/(0\sim 2\pi)$.

A \emph{vine} is a connected component of
\[
\mathcal{V}_{\ell}(\M,\gamma)\setminus (S^1\times \Delta).
\]
Equivalently, a vine is obtained by continuously following a single off-diagonal persistence point
as $t$ varies.
\end{definition}

We note that this definition can be used with either standard or extended persistence; in this work, we use extended persistence throughout for our later results.

For the purpose of this paper, we only need to define what it means for a vineyard \emph{not} to have monodromy or non-trivial topology. This is significantly simpler than defining monodromy and its order completely, as described in \cite{Chambers2026}. 

We first note that if a vine has a limit point on the diagonal, then it has two limit points on the diagonal, because by assumption there are only a finite number of points in the persistence diagram. This means that for every vine that starts and ends on the diagonal, we can define the \emph{extension} that connects the end point of the vine with the start point forward in periodic time. 

\begin{definition}\label{def:canonVineyard}
We say that a vine does not exhibit monodromy if its extension forms exactly a single cover of $S^1$. In other words, the extended vine intersects every time slice for each $t$  only once. 
We say that the vineyard does not exhibit monodromy if this holds for every vine in the vineyard. 

We say that a vineyard has trivial topology if there is an isotopy from $ S^1 \times \mathbb{R}^2$ to itself while mapping $S^1 \times \Delta$ to itself, such that the vines are mapped to either 
\begin{itemize}
    \item to exactly one of the $\{ (t , (i,i+1))| t \in [0,2\pi ] \}/\sim $, with $i \in \{ 1,\dots ,n\}$ where $n$ is the number of vines that do not have a limit point on the diagonal and $\sim$ indicates the identification of $0$ and $2 \pi$, if the (extended) vine has no limit point on the diagonal.  
    \item to exactly one of the  $\{ (t , (j,j + f(t) )|  t \in [0,2\pi ] \}/\sim $, with 
    \[f(t) = \begin{cases}
      t (1 -t) & \textrm{if } t \in [0,1]   \\
      0 & \textrm{otherwise}
    \end{cases}\ \] 
    and $j \in \{ n+1,\dots ,n+k\}$ where $k$ is the number of vines that does have a limit point on the diagonal and $\sim$ indicates the identification of $0$ and $2 \pi$,  if the (extended) vine has a limit point on the diagonal.  
\end{itemize} 
    See Figure~\ref{fig:trivial-vineyard} for an illustration.
\end{definition}

\begin{figure}[ht]
    \centering
    \includegraphics[width=0.5\linewidth]{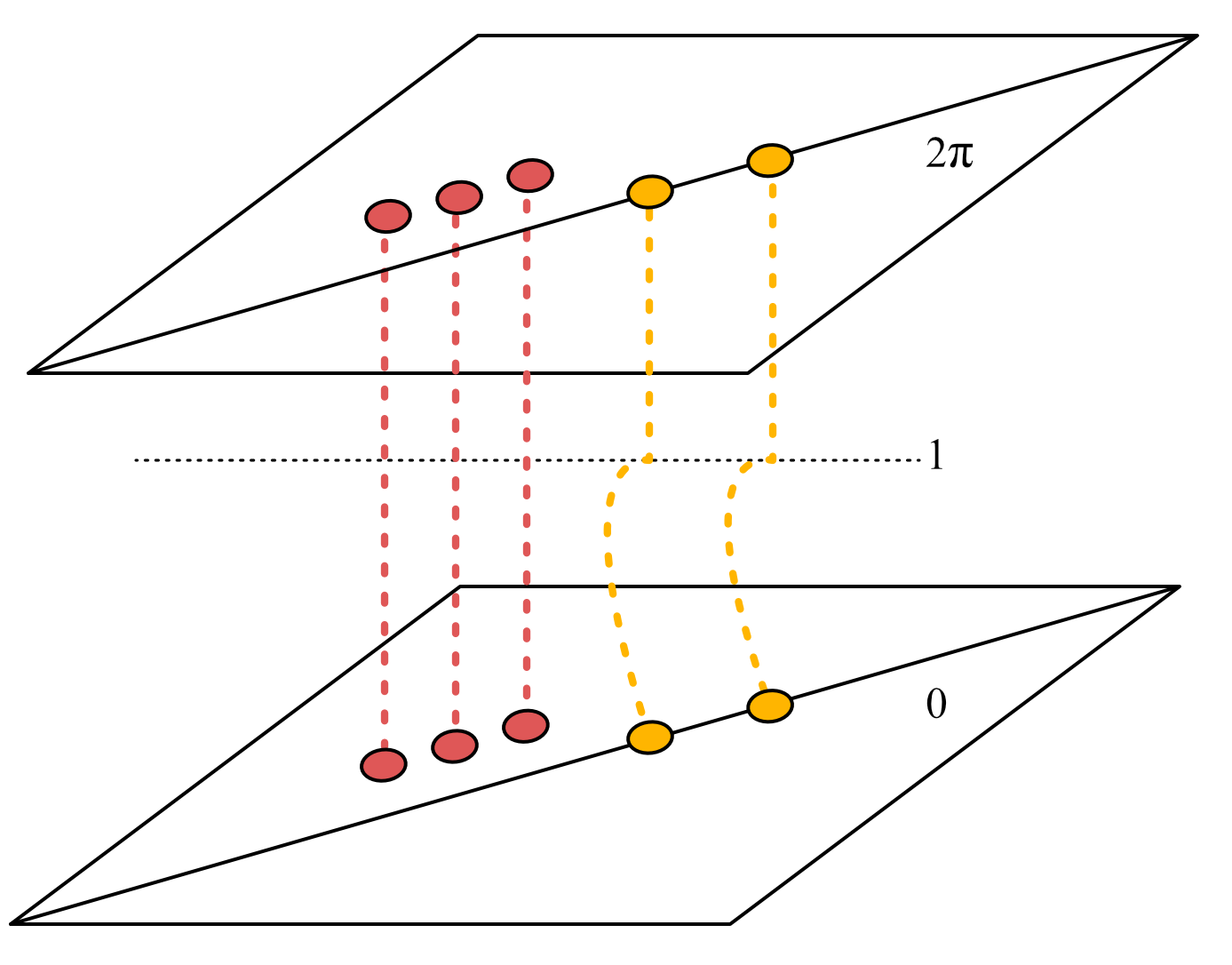} 
    \caption{  A canonical vineyard with trivial topology in the sense of Definition \ref{def:canonVineyard}.  }
        \label{fig:trivial-vineyard}
\end{figure}

We say that a vine induces an \emph{interchange} if it is a double cover of $S^1$ and following the points in a persistence diagram for a period of $2\pi$ interchanges those points.

\subsection{Singularities in the plane}
\label{sec:Singularities_Plane}
In this section we recall the classification of the singularities of the symmetry set and focal set. We call this union the generalized symmetry set; Bruce, Giblin, and Gibson \cite{Bruce1985symmetry}, and Giblin and Diatta \cite{giblin2003symmetry}\footnote{There is also a paper \cite{WrongVersionGiblin} with a similar title and significant overlap in content, where only Giblin is mentioned as an author. However, the part that is most useful to us can only be found in \cite{giblin2003symmetry}. }  refer to as the full bifurcation set. 
In this section we will assume that $\M$ is generic in the sense of \cite{Bruce1985symmetry}.

We will follow Arnold's notation for the singularities \cite{arnol1993singularity, arnol2003catastrophe,arnold1998singularity, Arnold2012}, to whose work we also refer the curious for general background reading. 
In two dimensions there are only singularities of type $A_k$. Specifically, the type $A_k$ denotes the setting where the Euclidean squared distance function $d_\mathbb{E} (\cdot,p)^2 |_\M$ restricted to the 1-dimensional manifold $\M$ is $\mathcal{R}$-equivalent (that is equivalent up to diffeomorphism, see \cite{bruce1992curves} for a full definition) to $\pm \tau^{k+1}+c$, where $\tau$ parametrizes $\M$ and $c$ indicates some constant (which can be interpreted as the radius squared). 
This means that $A_1$ indicates that there is a point where a circle centred on $p$ is tangent to $\M$. This is also referred to as an ordinary contact.  Similarly, $A_2$ indicates that there is a point on $\M$ such that $p$ is the centre of curvature of that point. Put differently, the circle centred at $p$ and $\M$ have the same first and second order derivatives at that point. We also say that there is a type $A_2$ contact. 
Finally, $A_3$ indicates that there is a point on $\M$ such that not only $p$ is a centre of curvature, but also the curvature at this point is also a local maximum (in absolute value). Put differently, the circle centred at $p$ and $\M$ have the same first, second, and third order derivatives at that point. Similarly to before, we say that there is a contact of type $A_3$.  
Different type contacts (of higher order) are not generic.

We write $A_i A_j$ if there are both a type $A_i$ and $A_j$ contact between the manifold and the same circle. This notation generalizes to an arbitrary number of contact points, here we allow ourselves to abbreviate where repeated contacts of the same type are indicated as a power, i.e. $A_1A_1= A_1^2$. If there are contacts between circles with different radii, but sharing the same centre, and the manifold (curve in this case) then we indicate this by a slash, for example $A_1^2/A_1^2$, indicates that there are two concentric circles, each of which with an ordinary contact with the curve $\M$. 

\begin{figure}[h!] 
    \centering
    \includegraphics[width=1\linewidth]{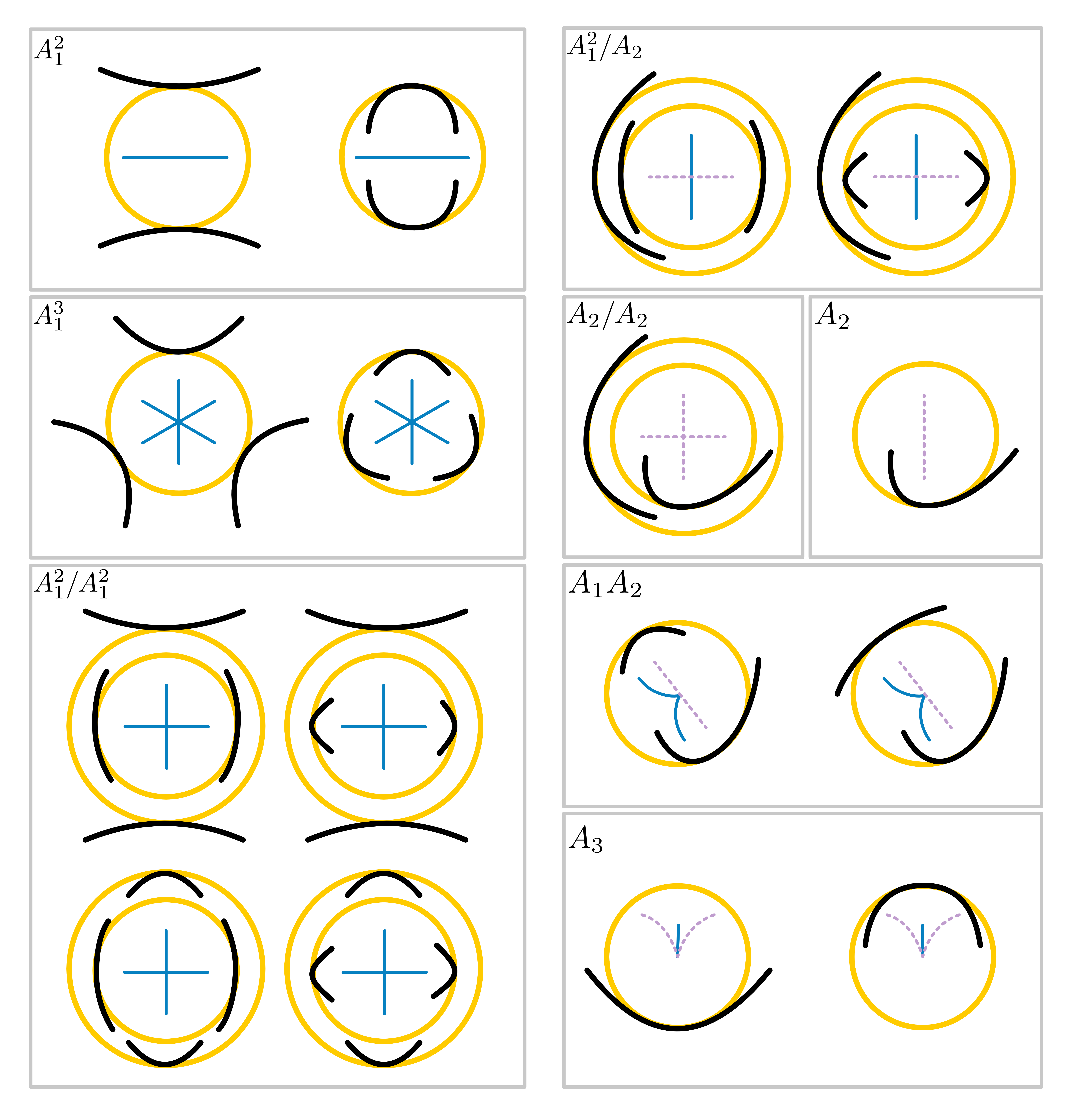}
    \caption{The singularities of the symmetry and focal set in the plane. We adopt Arnold's notation for these singularities. {We denote the symmetry set in blue, the focal set in dashed purple, the manifold in black, and the relevant circles in yellow. We simplify the drawing using straight lines when possible to represent the singularities clearly rather than forcing completely correct geometry.} {We have indicated a number of different cases for $A_1^2$, $A_3$, $A_1^2/A_2$, $A_1^2/A_1^2$, but not all (given the number of cases involved). Note that every contact point of type $A_1$ can be either a local minimum or maximum. }
    \label{fig:classification2D} } 
\end{figure}

Having now described the notation, we can go over the types of singularities. We note that Bruce, Giblin, and Gibson \cite{Bruce1985symmetry}, and Giblin and Diatta \cite{giblin2003symmetry} ignored the $A_2/A_2$ singularity because they were focused on the symmetry set, however, they were clearly aware of it because it shows up in their examples.  Their examples we will treat piecewise. We refer to Figure \ref{fig:classification2D} for an overview of the entire classification in two dimensions. 

\subparagraph{Type  \texorpdfstring{$A_1^2$}{A12}} 
The circle centred at $p$ is tangent to $\M$ at two different places. We stress that $d_\mathbb{E} (\cdot,p)^2 |_\M$ at the contact points can have a local minimum or local maximum, or geometrically speaking the curve $\M$ can locally lie inside or outside the circle.
The point $p$ lies on a manifold piece of the symmetry set, see Figure \ref{fig:classification2D}.  
We say that:
\begin{itemize} 
\item if the two contact points are local minima (a $0$-cycle is born for each of the Morse critical points in the persistence of the sublevel set) then $p$ lies on a birth-birth stratum of the symmetry set,
\item if the two contact points are local maxima (a $0$-cycle dies for each of the morse critical points in the persistence of the sublevel set) then $p$ lies on a death-death stratum of the symmetry set,
\item if one of the two contact points is a local minimum and the other a local maximum, then $p$ lies on a birth-death stratum of the symmetry set.
\end{itemize}



\subparagraph{Type  \texorpdfstring{$A_1^2/A_1^2$}{A12A12}}
Two circles of different radii centred at $p$ are both tangent to $\M$ at two different points (that is, there are $4$ contact points in total). This concerns a transversal intersection of manifold pieces of the symmetry set. Once again each of the contact points may correspond to either a local minimum or maximum.  

\subparagraph{Type  \texorpdfstring{$A_1^3$}{A13}}
A circle centred on $p$ is tangent to $\M$ at three different points. The symmetry set is locally diffeomorphic to three line pieces intersecting in a single point.  

\subparagraph{Type  \texorpdfstring{$A_2$}{A2}} The point $p$ is the centre of curvature of a point on $\M$, in other words, the circle centred at $p$ is an osculating circle to $\M$. The point $p$ lies on a manifold piece of the focal set (which is also called the evolute in the planar setting). 

\subparagraph{Type \texorpdfstring{$A_1A_2$}{A1A2}}
The circle centred on $p$ is tangent to $\M$ at two different points. One of these is an ordinary contact, that is, the circle is tangent to $\M$ there. The circle centred at $p$ is also a centre of curvature at a point of $\M$ (other than the point of tangency mentioned before), or, in other words, the circle has a second-order contact. The point $p$ is where a cusp of the symmetry set touches a manifold piece of the focal set, see Figure \ref{fig:classification2D}.  

\subparagraph{Type  \texorpdfstring{$A_3$}{A3}} 
In this case $p$ is not only a centre of curvature, but the curvature at the tangent point is a local minimum or maximum. The point $p$ lies on the closure of a symmetry set and on a cusp of the focal set.


\subparagraph{Type \texorpdfstring{$A_1^2/A_2$}{A1A2}}
Here, there are two circles that are concentric at $p$. The curve $\M$ is tangent to one of these circles in two places. The other concentric circle is an osculating circle of the manifold (in a place distinct from the previous two mentioned points). The point $p$ lies on a transverse intersection of a manifold piece of the symmetry set and a manifold piece of the focal set.

\subparagraph{Type \texorpdfstring{$A_2/A_2$}{A2A2}}
In this case, again, there are two concentric circles centred at $p$ (of different radii). Both circles are osculating circles to the curve/manifold $\M$. The point $p$ lies on the intersection of two manifold pieces of the focal set or evolute.

\subparagraph{Illustration} We provide illustrations of all generic singularities (which we have listed above) in an explicit example, as depicted in Figure \ref{fig:PotatoAnnotated}.  

\begin{figure}[h!]
    \centering
    \includegraphics[width=\linewidth]{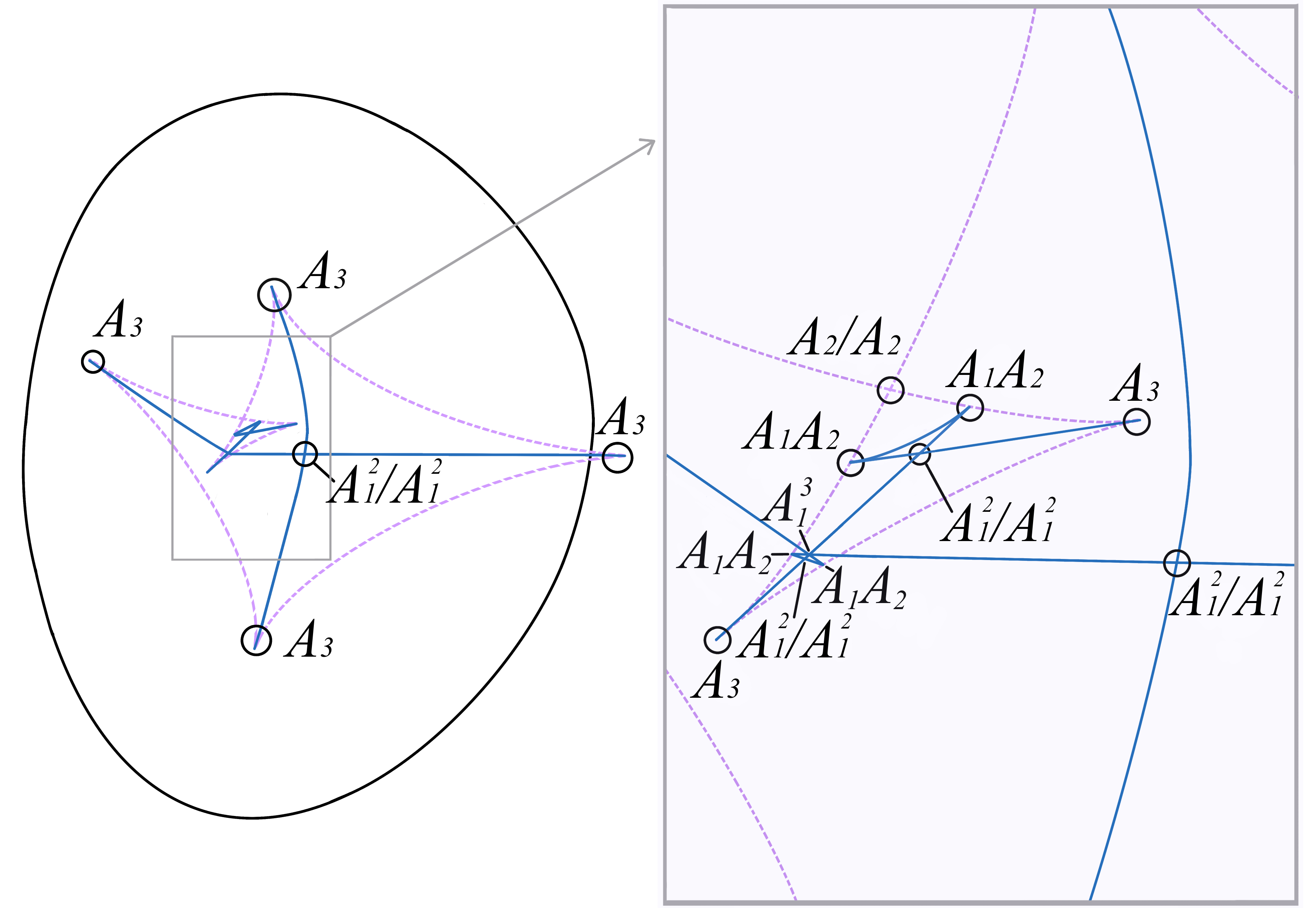}
    \caption{The singularities of the symmetry and focal set in the plane of a curve. Adjusted from Figure 3 of \cite{Bruce1985symmetry}, which was not annotated. On the right, one finds a zoom-in of the configuration of singularities in the centre of the configuration. 
    \label{fig:PotatoAnnotated} } 
\end{figure}

\subparagraph{Genericity assumptions} 
We note that the earlier assumption that the vineyard does not contain self-intersections corresponds to the avoidance of certain codimension $2$ or dimension $0$ singularities in the symmetry set. In fact, as we will discuss below, we will concentrate on closed curves or loops $\gamma(t)$ that are generic in  the sense that they avoid all codimension $2$ or dimension $0$ singularities in the symmetry set.

\subsection{Topological stability outside the symmetry set}

\begin{proposition}\label{prop:ConstOrd}
Let $\M$ be a manifold in $\mathbb{R}^2$ and $U$ be a connected component of the complement of $\operatorname{gsym}(\mathcal{M})$. Suppose that $y, y' \in U$ and let $(b_i,d_i)$, $(b_i',d_i')$ be the persistence points in $\operatorname{Dgm}(d_{\mathbb{E}}^2(x,y)\big|_{\mathcal{M}})$, $\operatorname{Dgm}(d_{\mathbb{E}}^2(x,y')\big|_{\mathcal{M}})$ respectively, where we order them according to birth time. This induces an ordering or permutation $\pi$ ($\pi'$, respectively) such that $d_{\pi(i)} \leq d_{\pi(i+1)}$ ($d_{\pi'(i)} \leq d_{\pi'(i+1)}$, respectively) for all $i$. We have that $\pi= \pi'$, and we say that the diagrams are combinatorially equivalent.  
\end{proposition}

\begin{proof}
Outside of the generalized symmetry set the function $d_{\mathbb{E}}^2(x,y)\big|_{\mathcal{M}}$ is Morse and has distinct critical values.\footnote{From the singularity theoretic viewpoint, this means that the squared distance function $d_\mathbb{E} (\cdot,p)^2 |_\M$ restricted to the 1-dimensional manifold $\M$ is $\mathcal{R}$-equivalent to either $\tau +c$ or $\pm \tau^2 +c$, where $\tau$ parametrizes $\M$ and $c$ indicates some constant. } Consider a curve $\gamma(t)$ which connects $y$ and $y'$, which exists because they lie in the same connected component. Now the order of the births and deaths in the persistence diagram $\operatorname{Dgm}(d_{\mathbb{E}}^2(x,\gamma(t) )\big|_{\mathcal{M}})$ cannot change with time, because for the order to change they would need to be equal at some point (by the intermediate value theorem), but this is excluded. 
\end{proof}

This means a vineyard has a simple structure if our observation loop avoids $\operatorname{gsym}(\mathcal{M})$: 
\begin{corollary} \label{Cor:NoMondromyInConnectedComponent}
    The vineyard of a loop $\gamma$ that is contained in such a connected component of the complement of $\operatorname{gsym}(\mathcal{M})$ is topologically the set of $k$ unlinked circles, where $k$ is the number of points in the persistence diagram $\operatorname{Dgm}(f_{\gamma(t^*) })$ for any and thus all $t^*$. Moreover the vineyard does not exhibit monodromy. 
\end{corollary}

\begin{corollary} We note that Proposition \ref{prop:ConstOrd} implies that the distinction between birth-birth, birth-death and death-death strata is well defined, in the sense that the point where you cross the same manifold piece of the symmetry set does not matter.
\end{corollary}

The previous can even be generalized: 
\begin{corollary} \label{cor:SmallLoops}
For the topology of the vineyard of $d_\mathbb{E} (\cdot ,\gamma(t))^2 |_\M$ it only matters which connected component of the complement of the generalized symmetry set $\gamma$ intersects and in which order (with respect to $t$), the precise geometry of $\gamma$ does not matter (as long as it is generic in the sense that it only crosses the focal and/or symmetry transversely and outside the singularities of codimension $2$). In particular, if we consider a loop enclosing a singularity, we can take that loop to be as small as we like. 
\end{corollary}

\section{The vineyard behaviour near the singularities in 2D 
}


In this section, we will investigate the behaviour of the vineyard induced by the family of sublevel filtrations of the function $d_\mathbb{E}^2(\cdot , \gamma(t) )|_\M$ for small observation loops $\gamma(t)$ that run around any $0$-dimensional singularity. By `any' we mean every type of singularity, which have been identified by Bruce, Giblin, and Gibson, and recalled in Section \ref{sec:Singularities_Plane}.
Additionally and in support of that goal, we will consider the behaviour for manifold pieces of the symmetry set and the focal set and also specifically describe the behaviour in the vineyard if $\gamma(t)$ is a curve (nonperiodic and the image of the unit interval~$[0,1]$).

\subparagraph{Interchanges} Before considering the topology of the vineyard of the function $d_\mathbb{E}^2(\cdot , \gamma(t) )|_\M$, we need to introduce some general nomenclature. 
We say that two points in a persistence diagram $\mathcal{V}_\ell(t^*)$ are \emph{interchanged} with period $2\pi$ if by following the vines $\mathcal{V} _1(t)$, $\mathcal{V}_2 (t)$ that go through these two points in the vineyard in the positive time direction for a duration of $2\pi$ the end point of the one vine is the starting point of the other, i.e. 
\begin{align}  
\mathcal{V} _1(t^*+ 2\pi) &= \mathcal{V}_2 (t^*) & 
\mathcal{V} _2(t^*+2\pi) &= \mathcal{V}_1 (t^*).
\nonumber
\end{align} 
Below we will see that such interchanges can be realized, see Figure \ref{fig:figure3ForA12A12Transition}. 
More generally, following the vines for a period (of $2 \pi$) induces a discrete map or relation between the points in the persistence diagram and the diagonal, mapping the points in the persistence diagram at time $0$ (or diagonal) to the points in the persistence diagram at time $2\pi$ (or the diagonal). 
If the vineyard stays away from the diagonal (and is generic, as we assumed in general), this map is a bijection, while generally many points can be mapped to the diagonal and as many can emanate from the diagonal making this a relation. 

\subparagraph{Interchanges and the symmetry set} 
Returning to the specific setting of $d_\mathbb{E}(\cdot , \gamma(t) )^2|_\M$, we make the following observation: 
\begin{observation}\label{ob:$A_1^2$ intersection} Assuming that the vineyard stays clear from the diagonal (that is $\gamma(t)$ does not cross the focal set), if there is an interchange between two points in the vineyards $d_\mathbb{E}(\cdot , \gamma(t) )|_\M$, then $\gamma(t)$ intersects at least two different strata of the symmetry set (which does not include the focal set), which correspond to a birth-birth and a death-death $A_1^2$ contact respectively. Here the strata are the connected components of the generalized symmetry set minus its codimension $2$ or $0$ dimensional singularities. 
\end{observation}


\begin{proof} 
Let $\mathcal{V} _1(t)$, $\mathcal{V}_2 (t)$ be the two vines that are involved in the interchange. We then consider the projection of $\mathcal{V} _1(t)$ and $\mathcal{V}_2 (t)$ to the birth and death coordinates. The interchange of the points in the persistence diagram implies in particular that the birth and death values of $\mathcal{V} _1(t)$ and $\mathcal{V}_2 (t)$ interchange on the interval $[0,2\pi]$. By the intermediate value theorem \cite{rudin1976principles} 
there have to be times $t_1$ and $t_2$ for which the birth, respectively death values coincide. The point $\gamma(t_1)$ ($\gamma(t_2)$) lies on a birth-birth stratum 
(death-death stratum respectively) of the symmetry set. 
\end{proof} 

\begin{corollary}
\label{cor:interchange} Note that this immediately implies that at least $4$ contact points are involved in an interchange. Moreover, these contact points should be associated to the two persistence pairs involved in the interchange.
\end{corollary}

We can make two further observations that will be critical in our analysis below. 
\begin{observation}\label{obs:Isolation}
We know that in the generic setting the dimension $0$ singularities are isolated (in the sense that there is a lower bound on the distance from a given singularity to any other $0$-dimensional singularity), this can be seen explicitly in the classification that was recalled in Section \ref{sec:Singularities_Plane}. Moreover, we have that persistence diagrams are Lipschitz with respect to the input, see \cite{CohenSteinerStability}. This means that if we consider a small loop $\gamma(t)$, that is contained in a small ball, where here small means small compared to the distance between singularities and the Lipschitz constant, and $\gamma(t)$ crosses the focal set, then the points in the persistence diagram that come out of (or go into) the diagonal when doing so will have less persistence than all other points in the persistence diagram. This means that if $\gamma(t)$ is sufficiently small, the extended vine with the small persistence is separate from the other vines, that is, it is unlinked. 

We stress that it is possible for two persistence points to come out of (or go into) the diagonal near a singularity of type $A_2/A_2$.  In this case there are two vines with low persistence but with different birth/death times (which are almost equal). More precisely we can assume that the difference in birth/death times are significantly smaller than the persistence.  All other points in the persistence diagram can be assumed to have much larger persistence so the corresponding vines are unlinked from each other and the other vines. 
\end{observation}
In generalization of Proposition \ref{prop:ConstOrd} and Corollary \ref{Cor:NoMondromyInConnectedComponent}, we have the following:  
\begin{observation}\label{Obs:NoLinkingSymmetry}
Let $\gamma(t)$ be a loop and let $\mathcal{V}_i$ be vines and thus $\mathcal{V}_i(t)$ be points in the persistence diagram, where $i \in I$ indices the vines. We say that $\mathcal{V}_i (t)$ is not associated to any of the crossings of the symmetry set or focal set of $\gamma$, if the projections of $\mathcal{V}_i (t)$ to the birth and death values give unique values and are distinct from all other birth and death values of $\mathcal{V}_j (t)$, with $i \neq j$. If this holds for all $i \in I$, we say that all the vines in this set are not associated to the generalized symmetry set. We note that as in Proposition \ref{prop:ConstOrd} these points $\mathcal{V}_i (t)$, $i \in I$ are combinatorially equivalent. And as in Corollary \ref{Cor:NoMondromyInConnectedComponent}, these vines form unlinked circles and the vines do not exhibit monodromy. 

More generally, we say that $\mathcal{V}_i (t)$ is not associated to any of the birth (respectively death) crossings of the symmetry set or focal set of $\gamma$, if the projections of $\mathcal{V}_i (t)$ to the birth and death values give unique values and distinct from all other birth (respectively death) values of $\mathcal{V}_j (t)$, with $i \neq j$.  If this holds for all $i \in I'$ with $I' \subseteq I$ a subindex set, we say that all the vines in this subset are not associated to the birth (respectively death) part of the generalized symmetry set. This means that projection of the $\mathcal{V}_j (t)$ to the birth (respectively)  death values are distinct for all $t$ and thus the vines are unlinked and there is no monodromy.
\end{observation}


\subsection{Type \texorpdfstring{$A_1^2$}{A12}} 

As we have seen in Section \ref{sec:Singularities_Plane}, if $p$ lies on a singularity of type $A_1^2$ on the symmetry set this means that there is a circle centred at $p$ that is tangent to $\M$ at two different points, i.e. there are two contact points. Moreover, there is no other circle centred at $p$ that has either multiple tangencies, or higher order ($A_2$ or $A_3$) contact points. 
The two contact points on the circle can each correspond to either a birth or a death of a cycle. 
As mentioned in Section \ref{sec:Singularities_Plane}, we refer to this part of the symmetry set or stratum as a birth-birth, death-death or birth-death stratum or piece depending on the nature of the critical points. 
With a little abuse of notation we will write $C^1$, $C^2$ 
for both the critical values and the critical points, and in case we specify if it is a birth-birth, death-death, or birth-death stratum, we will denote both the critical points and the values by $b^1$ and $b^2$, $D^{I}$ and $D^{II}$, and $b^1$ and $D^{II}$, respectively. If the critical values and, in particular, the birth or death values or the associated critical points depend on a parameter $t$ this is indicated by adding $(t)$, that is $C_1(t)$, $b^{1}(t)$, $D^{I}(t)$, etc.

We now first consider a small/short smooth curve $\gamma(t)$, which crosses the symmetry set transversally. As $\gamma (t)$ crosses the symmetry set, there are two things that happen or can happen to the sublevel sets of $d_\mathbb{E} (\cdot , \gamma(t) ) |_\M$:
\begin{itemize} 
\item \textbf{Reordering.} 
It always happens that critical values $C^1 (t)$, $C^2(t)$ interchange order, which in topological terms means that the births or deaths of the cycles occur in a different order than before. 
\item \textbf{Elder rule.} Due to the elder rule, it is possible that the pairing between the critical points into persistence pairs changes. For example, if before the crossing the symmetry set at a birth-birth stratum, the persistence points $(b^1(t) , D^K(t) )$, $(b^2(t) , D^L(t))$ are in the persistence diagram, then after the crossing, the pairing between births and deaths may yield $(b^2(t) , D^K(t) )$, $(b^1(t) , D^L(t))$, where $D^K(t)$ and $D^L(t)$ are deaths uninvolved in the crossing. 
\end{itemize} 

The involvement of the elder rule leads to so-called knees \cite{CohenSteiner2006}, and we note that these knees play an essential role in the higher order medial axes and Faustian interchanges of \cite{edelsbrunner2025mid, ElizabethThesis}.

We make two additional observations: 
Firstly, the second case 
requires very specific global configurations. 
Secondly, the second case, gives rise to knees, that is, points where the vines in the vineyard are not smooth. These have been classified in~\cite{CohenSteiner2006}.  

Now we consider a (small) loop $\gamma(t)$, centred on $p$ (the point on the symmetry set, mentioned before), which therefore intersects the symmetry set twice, transversally and in opposite directions. Because we transverse the same stratum of the symmetry set in opposite directions, both change the reordering and the elder rule change (if the latter occurs) will be undone as you cross the symmetry set for the second time. 
This means in particular that there is no monodromy when you follow the loop $\gamma(t)$, which we already knew because the conditions of Observation \ref{ob:$A_1^2$ intersection} and Corollary \ref{cor:interchange} are not satisfied. We note also that the vines are unlinked, thanks to Observation \ref{Obs:NoLinkingSymmetry}. 
We refer to Figure \ref{fig:A_12 Singularity}.


\begin{figure}[h!]
\centering
\includegraphics[width=0.80\linewidth]{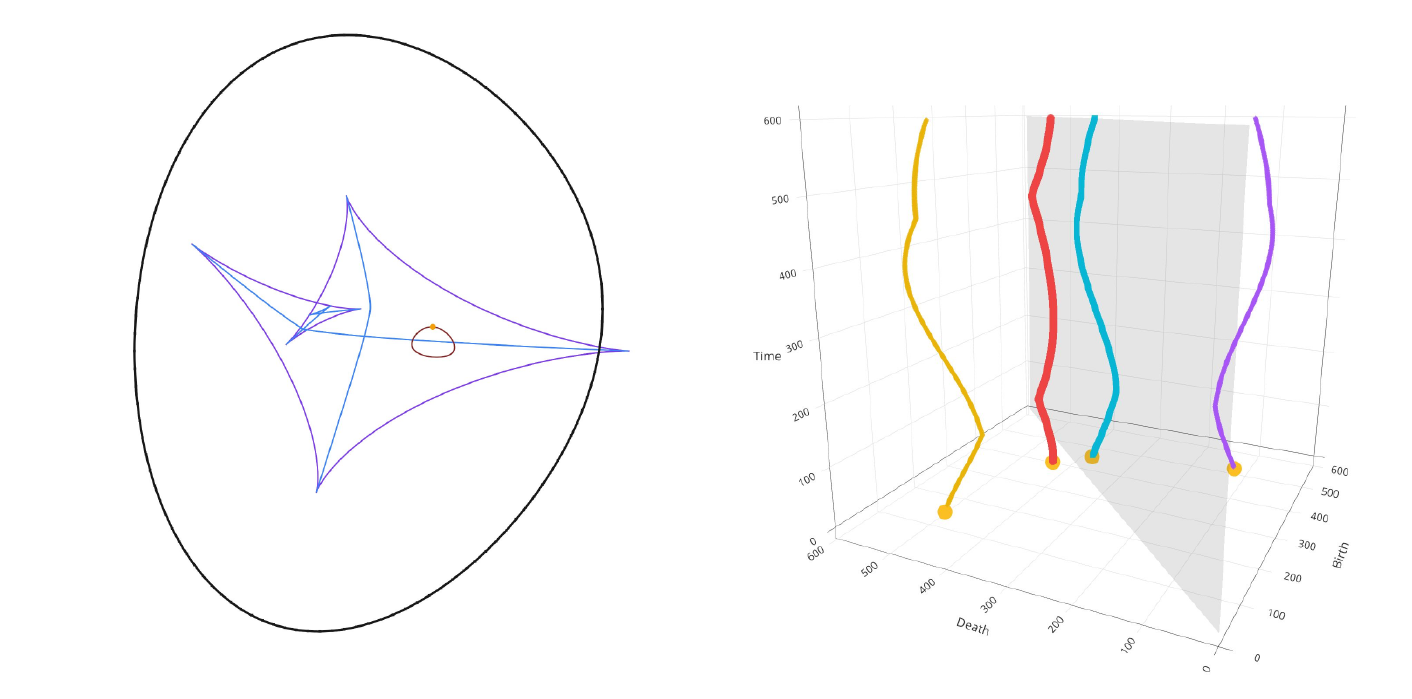}

\caption{A loop $\gamma$ enclosing the symmetry set or an $A_1^2$ singularity (left)  and the associated vineyard (right). }
\label{fig:A_12 Singularity}
\end{figure}

\subsection{Type \texorpdfstring{$A_1^2/A_1^2$}{A12A12}} \label{sec:A12A12}
As mentioned in Section \ref{sec:Singularities_Plane}, this is a transversal intersection point $p$ of two manifold pieces of the symmetry set. We will refer to the $4$ manifold pieces of the symmetry set that arise by removing the intersection point $p$ as the $4$ segments. As these $4$ pieces are each of type $A_1^2$, we can find the behaviour described above for each of them. 

The $4$ contact points on the two circles can each correspond to either a local minimum or a local maximum of the distance function $d_\mathbb{E} (\cdot,p)^2 |_\M$. At the level of persistence, this means that the contacts can each correspond to either a birth or a death of a cycle. 

We will now concentrate on exhibiting that singularities of type $A_1^2/A_1^2$ can indeed induce monodromy: 
In order to create an interchange, the two contact points on the smaller of the two concentric circles need to correspond to births and the two contact points on the larger of the two concentric circles need to correspond to deaths, because of Corollary \ref{cor:interchange}. In other words, in the case of an interchange, the symmetry set is locally a transversal crossing of a birth-birth and death-death manifold piece, see Figure \ref{fig:figure1ForA12A12Transition}.  

\begin{figure}[h!]
    \centering
    \includegraphics[width=0.65\linewidth]{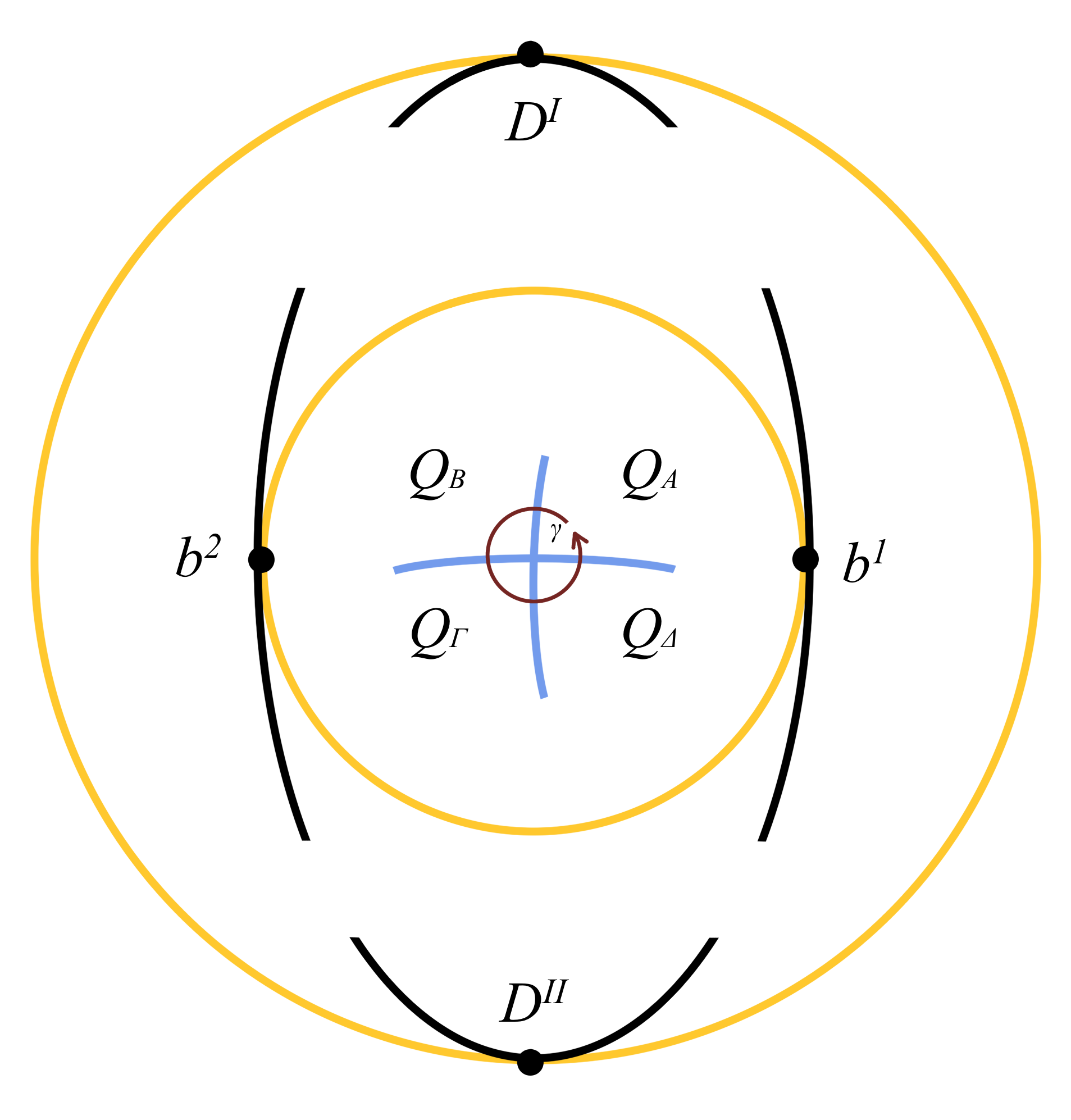}
    \caption{Some necessary conditions for the configuration of $\M$ (in yellow) with an interchange. The observation loop is indicated in blue the different quadrants with respect to the local structure of the symmetry set are indicated in reddish brown. 
    \label{fig:figure1ForA12A12Transition} } 
\end{figure}

We'll now consider a small loop $\gamma(t)$ that encircles the $A_1^2/A_1^2$ singularity. Near the $A_1^2/A_1^2$ singularity, the symmetry set divides the space into 4 connected components, which, because the self-crossing of the symmetry set near the $A_1^2/A_1^2$ singularity is transversal, we'll refer to as quadrants, called the $\Alpha$, $\Beta$, $\Gamma$ and $\Delta$ quadrants and denoted by $Q_\Alpha$, $Q_\Beta$, $Q_\Gamma$, and $Q_\Delta$. With a little abuse of notation, we refer to the contact points corresponding to births by $b^1$ and $b^2$ and use the same notation for the value of the function (that is, the radius). Similarly, we refer to both the contact points and function values corresponding to death by $D^I$ and $D^{II}$. 

Again, thanks to Corollary \ref{cor:interchange} the births at $b^1$ and $b^2$ and deaths at $D^I$ and $D^{II}$ need to be paired into persistence points $(b^i,D^J)$, and these values cannot be paired with other critical points. We stress that this does not mean that the geometric pairing of the points needs to be constant: Thanks to the elder rule, the pairing can change upon crossing the symmetry set, as we have seen in the discussion of the (co)dimension $1$ singularity $A_1^2$. This will in fact play an important role in the rest of the discussion. 


If we now return to the geometry we observe that in the $\Alpha$ quadrant the cycle born at $b^1$ (or more precisely there is a cycle born very close to $b^1$) is born before the cycle born at $b^2$  (or more precisely there is a cycle born very close to $b^2$), assuming the configuration is as depicted in Figure \ref{fig:figure1ForA12A12Transition}. We will ignore the fact that the contact point moves a little with the point on observation loop $\gamma(t)$, and just refer to $b^i$, $D^J$ respectively. Moreover, in the $\Alpha$ quadrant the cycle that dies at $D^I$ dies before the cycle that dies at $D^{II}$.
In the $\Beta$ quadrant $b^2$ is born before $b^1$, that is, the order changes upon crossing the symmetry set between the $\Alpha$~and~$\Beta$ quadrants. In the $\Beta$ quadrant the cycle that dies at $D^I$ dies before the cycle that dies at $D^{II}$ (as is the case for the $\Alpha$ quadrant). Upon crossing to the $\Gamma$ quadrant order in which $D^{I}$ and $D^{II}$ die is reversed, i.e. $D^{II}$ dies before $D^{I}$. The order of the births in the $\Gamma$ quadrant is the same as the order in the $\Beta$ quadrant, that is $b^2$ is born before $b^1$. 
In the transition from the $\Gamma$ quadrant to the $\Delta$ quadrant the order of births is interchanged from $b^2$ before $b^1$ to $b^1$ before $b^2$, while the death order remains invariant, i.e. $D^{II}$ dies before $D^{I}$. Finally in the transition back from the $\Delta$ quadrant to the $\Alpha$ quadrant the order of the deaths are again reversed, that is in the $\Delta$ quadrant $D^{II}$ dies before $D^{I}$, while in the $\Alpha$ quadrant $D^{I}$ dies before $D^{II}$. 

Roughly speaking, if we parametrize $\gamma(t)$ by $(\eta  \cos(t),\eta \sin (t))$, with $\eta$ a small constant and where we assume that the $A_1^2/A_1^2$ singularity is at the origin, we have up to reparametrization of the time (which we'll in this case view as an element of $\mathbb{S}^1$) and up to leading order in $\eta$ that 
\begin{align} 
b^1 (t) &\simeq r_1- L_1\eta \cos(t)
& b^2 (t) &\simeq r_1+ L_2 \eta \cos(t)
\nonumber
\\ D^{I} (t) &\simeq r_2- L_3 \eta \sin(t)
& D^{II} (t) &\simeq r_2+ L_4 \eta \sin(t)
\nonumber
\end{align}
where $r_1$ is the radius of the smaller concentric circle and $r_2$ is the radius of the larger concentric circle and $L_i>0$, with $i= 1,\dots, 4$, are (Lipschitz) constants.\footnote{The Lipschitz constants exist due to the fact that outside the focal set, 
the natural map from the normal bundle to Euclidean space is a local diffeomorphism, but whose differential is not (generally) equal to the identity, c.~f.~\cite[Section 10.4]{DoCarmoRiemannian}.}

Now we need to consider the pairing of the births and deaths. If the pairing is constant throughout, that is, the relevant points in the persistence diagram are $(b^1(t),D^{I}(t))$ and $(b^2(t),D^{II}(t))$ or $(b^1(t),D^{II}(t))$ and $(b^2(t),D^{I}(t))$, then this does not lead to monodromy, as following the vine in the vineyard for a period $2\pi$ will give the same as the starting point. 
However, because $D^{I} = D^{II}$ at the transitions between the $\Beta$ and $\Gamma$ quadrants and the $\Delta$ quadrant and $\Alpha$ quadrants, in principle, the pairing could change thanks to the elder rule. Similarly, because $b^1 = b^2$ at the transitions between the $\Alpha$ and $\Beta$ quadrants and the $\Gamma$ quadrant and $\Delta$ quadrants, in principle, the pairing could change thanks to the elder rule. 

\begin{figure}[ht!]
    \centering
    \includegraphics[width=0.6\linewidth]{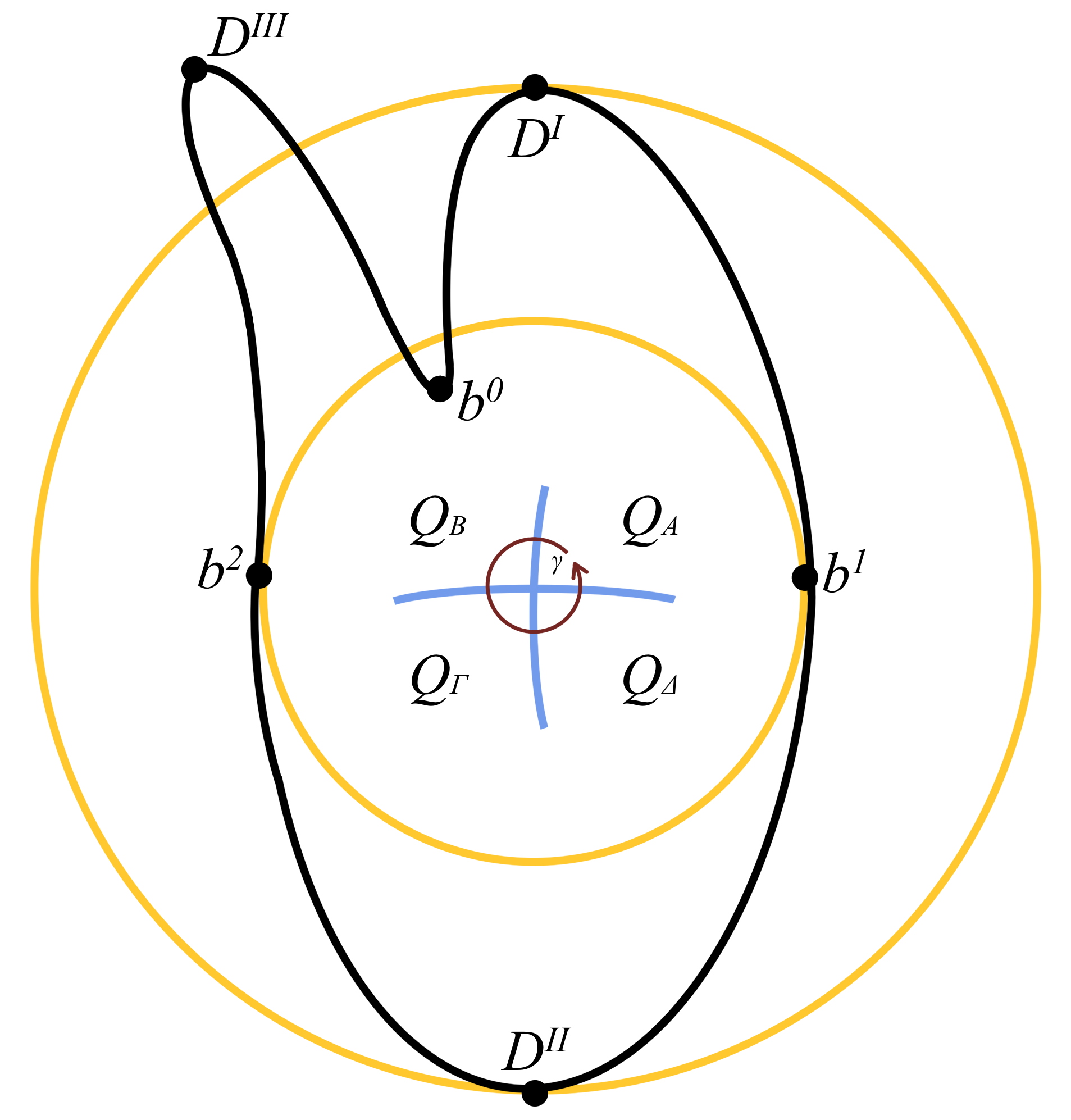}
    \caption{The part of the symmetry set depicted in the figure is not to scale. 
    \label{fig:figure2ForA12A12Transition} } 
\end{figure}

The change of pairing means that the points in the persistence diagram start moving in the opposite direction. To put it differently, with the pairing  $(b^1(t),D^{I}(t))$ and $(b^2(t),D^{II}(t))$ the points in the persistence diagram move (up to translation and up to leading order in $\eta$) as 
\begin{align}
&(- L_1 \eta \cos(t),- L_3 \eta \sin(t)) & & (L_2 \eta \cos(t), L_4 \eta \sin(t)), 
\nonumber
\end{align} 
that is counter clockwise,
while with the pairing $(b^1(t),D^{II}(t))$ and $(b^2(t),D^{I}(t))$
the points in the persistence diagram move (up to translation) as 
\begin{align}
&(- L_1 \eta \cos(-t),- L_4\eta \sin(-t)) & & ( L_2 \eta \cos(-t), L_3 \eta \sin(-t)), 
\nonumber
\end{align} 
where we used that $\cos(t) =\cos(-t)$, that is clockwise. Hence at the change of a pairing the direction of the points in the persistence diagram changes. 
The transitions between these pairings give rise to knees in the vineyard.

We now observe that without a change of direction, the points in the persistence diagram end up in their original position after a period of $2\pi$, while at the same time changing direction $4$ times will also mean that the points end up in their original position. The only remaining possibility for monodromy to occur is for two consecutive changes of quadrant to both yield a change of direction/pairing. 

This leads us to the following definition:
\begin{definition}[Monodromy critical $A_1^2/A_1^2$ singularities]
We call a singularity of $A_1^2/A_1^2$ monodromy critical if there are two consecutive changes in the pairing of critical points (induced by the elder rule), if the loop encircles an $A_1^2/A_1^2$ singularity.
\end{definition}

\begin{figure}[h!]
    \centering
    \includegraphics[width=0.99\linewidth]{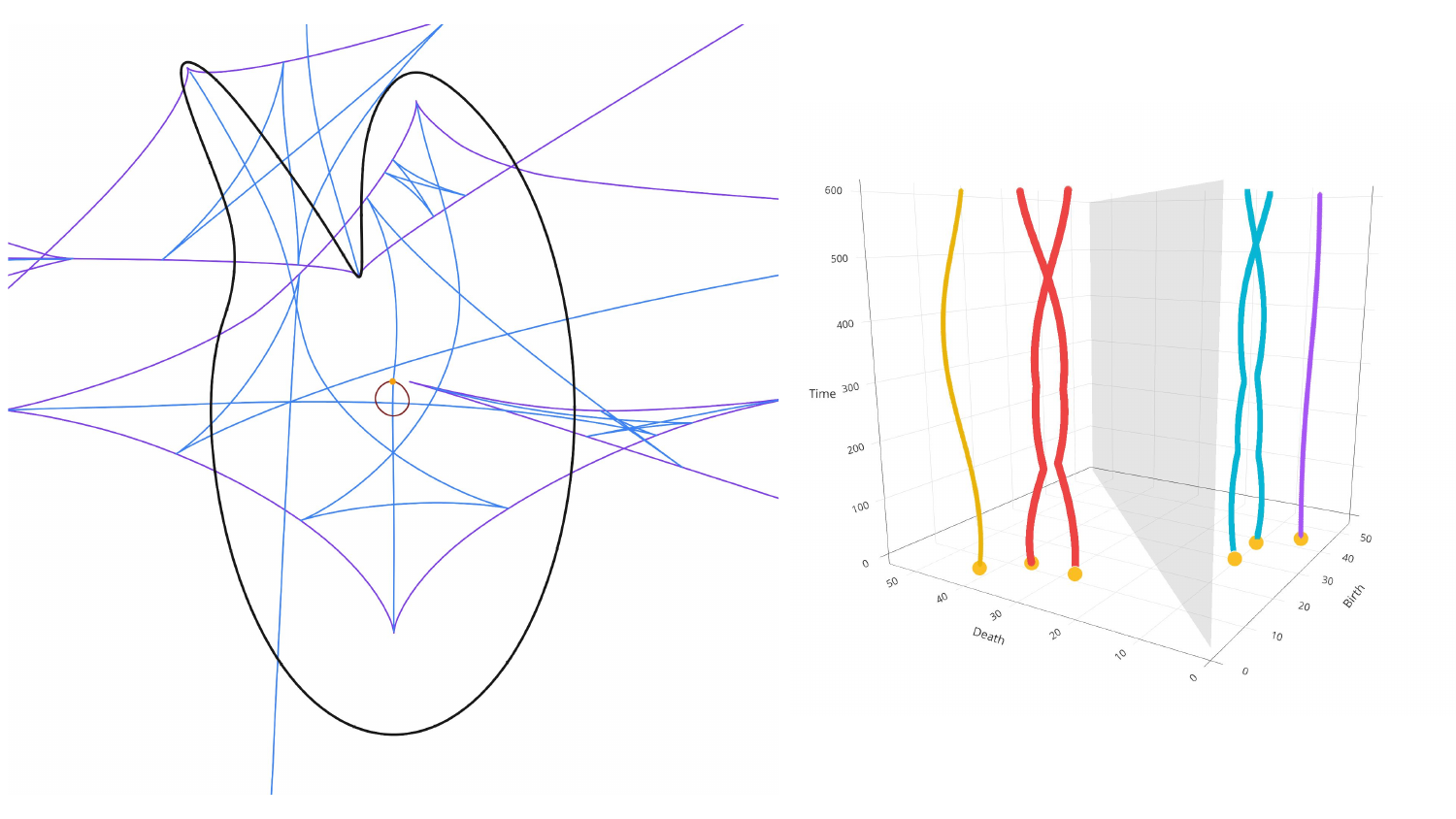}
    \caption{Left: The complete symmetry (blue) and focal (purple) sets for the example where a single singularity of type $A_1^2/A_1^2$ generates monodromy. Right: The vineyard exhibiting the monodromy. We stress that the red and turquoise vines do not intersect in three dimensions (the intersection is a result of the projection to the two dimensional paper). 
    }
    \label{fig:figure3ForA12A12Transition} 
\end{figure}

Let's now discuss the example given in Figures \ref{fig:figure2ForA12A12Transition} and \ref{fig:figure3ForA12A12Transition}. The first birth occurs at $b^0$ and the final death at $D^{III}$ (again we ignore minor shifts in the contact points as we move along $\gamma$). Thanks to the elder rule this means that 
all other critical points need to be paired with each other. 
We go through the changes per quadrant: 
\begin{itemize} 
\item In $Q_\Alpha$, because $D^{I} <D^{II}$, the connected component created at $b^1$ merges with the one created at $b^0$ before $b^2$ merges with the resulting sublevel set as $D^{II}$, that is, $b^1$ is paired with $D^{I}$ and $b^2$ with $D^{II}$. 
\item In $Q_\Beta$, the pairing into a persistence point is the same as in $Q_\Alpha$: Because $D^{I} <D^{II}$, the connected component created at $b^1$ merges with the one created at $b^0$ before $b^2$ merges with with the resulting sublevel set as $D^{II}$, that is, $b^1$ is paired with $D^{I}$ and $b^2$ with $D^{II}$. 
\item Moving to $Q_\Gamma$ changes the situation: In $Q_\Gamma$, we have that $b^2 < b^1$ and $D^{II} <D^{I}$. Because $D^{II} <D^{I}$, the cycles born at $b^1$ and $b^2$ merge with each other before merging with $b^0$ at $D^I$. Since the cycle born at $b^2$ is older than the one born at $b^1$, we have the pairings $(b^2 ,D^{I})$ and $(b^1, D^{II})$. This means that there is a change of pairing compared to $Q_\Beta$. 
\item In $Q_\Delta$, we have that $b^1 < b^2$ and $D^{II} <D^{I}$. It is still true (compared to the situation in $Q_\Gamma$) that because $D^{II} <D^{I}$, the cycles born at $b^1$ and $b^2$ merge with each other before merging with $b^0$ at $D^I$. However, now the cycle born at $b^1$ is older than the one born at $b^2$ and the elder rule implies that $b^1$ is paired with $D^{I}$ This means that we have the pairings $(b^1 ,D^{I})$ and $(b^2, D^{II})$.
\end{itemize}
We stress that coming back from the delta quadrant to the alpha quadrant there is no change in pairing. 
In summary we have the following pairing:
\begin{center}
\begin{tabular}{ c |c  }
Quadrant  & pairing \\ 
\hline
$\Alpha$ & $(b^1 ,D^{I})$, $(b^2, D^{II})$ \\  
$\Beta$ & $(b^1 ,D^{I})$, $(b^2, D^{II})$
\\
$\Gamma$ & $(b^2 ,D^{I})$, $(b^1, D^{II})$
\\
$\Delta$ & $(b^1 ,D^{I})$, $(b^2, D^{II})$
\end{tabular}
\end{center}
This means that following the vines induces an interchange of the points in the persistence diagram; the points themselves as a set in $\mathbb{R}^2$ vary continuously, but following a single point across time yields an interchange. 




\begin{remark} The example in this section can be easily extended to arbitrary dimensions by taking an offset creating a tubular neighbourhood of the curve of `torus' $\mathbb{S}^1 \times \mathbb{S}^l$.
\end{remark}

\begin{remark} 
In our description we have focused on ordinary persistence, or the part of the extended persistence diagram that is associated with the `up' direction of the extended persistence. Due to symmetry, the behaviour of these vines (which lie above the diagonal) is reflected below the diagonal for the extended persistence diagrams~\cite{CohenSteiner2008,Turner2024}. We depict the vineyards in extended persistence diagrams in Figures \ref{fig:figure3ForA12A12Transition}, \ref{fig:A_13 Singularity}, \ref{fig:A_2 Singularity}, \ref{fig:A_1A_2 Singularity}, \ref{fig:A_3 Singularity}, \ref{fig:A_2/A_2 Singularity}, \ref{fig:A12/A_2 Singularity}, and  \ref{fig:a12a12_bigloop}. These figures were made using the Bouquet2D tool \cite{Bouquet}. 
\end{remark} 

\begin{remark} 
The spiral, which was the example in which monodromy was shown by Arya et al. \cite{Arya2024} can be deformed in such a way that a type $A_1^2/A_1^2$ occurs in the centre and the pairing of the births and deaths is identical to the one given in the table above. This is explained pictorially in Figure \ref{fig:ellipse_spiral}; note that this example depicts a manifold with boundary, which formally violates the assumptions. However, a small generic offset (with slightly non-uniform distance to preserve genericity) produces a manifold whose singular structure is arbitrarily close to—and contains—the structure shown here. This offset manifold exhibits a “doubling” phenomenon: each branch of the generalized symmetry set appears once for the interior and once for the exterior, together with additional structure near the endpoints. Since displaying the full structure makes the figure difficult to interpret, we use the simplified depiction shown here.
\end{remark}

\begin{figure}[h!]
\centering
\includegraphics[width=0.99\linewidth]{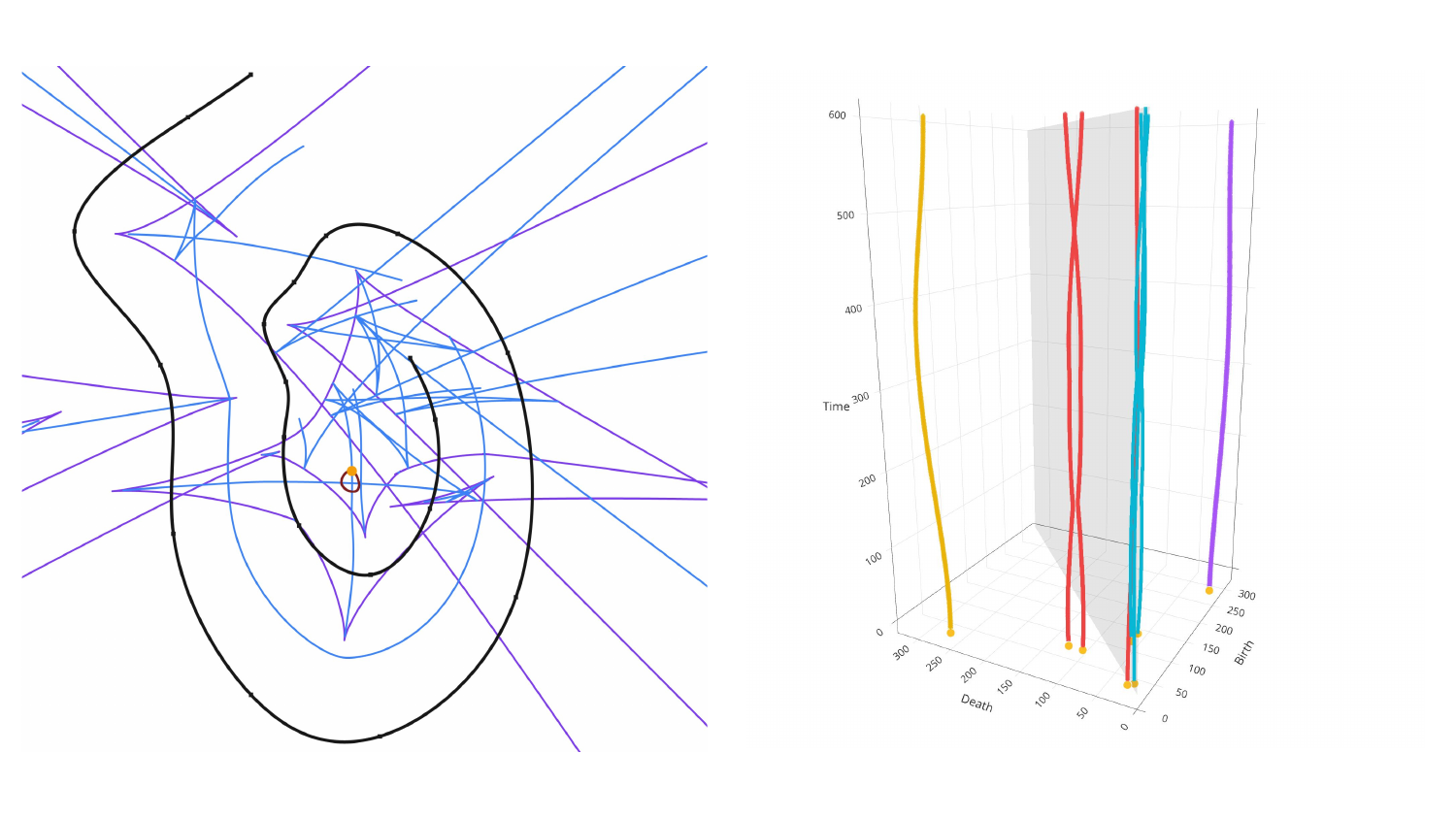}

\caption{Left: An example of a deformed (ellipsoidal) spiral inspired from the example by \cite{Arya2024}. The dark red observation loop surrounding the singularity of type $A_1^2/A_1^2$ yields monodromy. Right: The vineyard demonstrating monodromy along the loop containing the relevant $A_1^2/A_1^2$ singularity. }
\label{fig:ellipse_spiral}
\end{figure}



\subsection{Type \texorpdfstring{$A_1^3$}{A13}}


As we have seen in Section \ref{sec:Singularities_Plane}, if $p$ lies on a singularity of type $A_1^3$ on the symmetry set, this means that there is a circle centred at $p$ that is tangent to $\M$ at three different points, i.e.\ there are three contact points. Moreover, there is no other circle centred at $p$ that has either additional tangencies, or higher order ($A_2$ or $A_3$) contact points. The three contact points on the circle can, in general, correspond to either births or deaths. We denote birth events as $b^i$ and death events by $D^J$. By Observation~\ref{ob:$A_1^2$ intersection} and Corollary~\ref{cor:interchange}, any nontrivial interchange along a loop must intersect two distinct strata of the symmetry set (one birth–birth and one death–death) and therefore involve at least four contact points associated to exactly two cycles. An $A_1^3$ neighbourhood, however, supplies only three simultaneous $A_1$ contacts on a single circle (hence at most three involved contact points), so the necessary condition is not met and no nontrivial monodromy can occur.
To see that there is no non-trivial topology in the vineyard we have to distinguish a number of cases. If the contact points on the circle correspond to either all births or all deaths, then there is no linkage due to Observation \ref{Obs:NoLinkingSymmetry}. So let us consider the case where we have two deaths and one birth. The other case (with two births and one death) is symmetric. We'll denote these two deaths by $D^I$ and $D^{II}$ and the birth by $b^3$. Clearly, because we are away from the focal set, $b^3$ cannot be paired in to a persistence pair with $D^{I}$ or $D^{II}$, therefore there exist $b^1,b^2<D^{I}, D^{II}$ and $b^3<D^{III}$, such that $b^1$, $b^2$ are paired with $D^{I}$ or $D^{II}$ and $b^3$ with $D^{III}$. The former pairing does not have to be stable; depending on the elder rule the pairing can switch upon crossing the symmetry set. Thanks to Corollary \ref{cor:SmallLoops} we can assume that the loop $\gamma$ is arbitrary small and in particular much smaller than the difference between any of the values of births $b^j$ and deaths $D^J$ (where we mean both differences between pairs or births/deaths and differences between a birth and a death), with the exception of $D^I$, $D^{II}$, and $b^3$. However, the births $b^1$ and $b^2$ are distinct, and the same holds for $D^{III}$. This means that in the persistence diagram $\mathcal{V}(\tau)$, where $\tau$ is any fixed time, the distance between the points in the persistence diagrams is lower bounded. On the other hand, because $\gamma$ can be chosen arbitrarily small, and, thanks to the stability of persistence diagrams, the persistence points in $\mathcal{V}(t)$ can be assumed to be arbitrarily close to the points in $\mathcal{V}(\tau)$. 
Hence, we can conclude that the vines are unlinked. 
See Figure \ref{fig:A13PersistenceDiagram} for a sketch.

\begin{figure}[h!]
\centering
\includegraphics[width=0.5\linewidth]{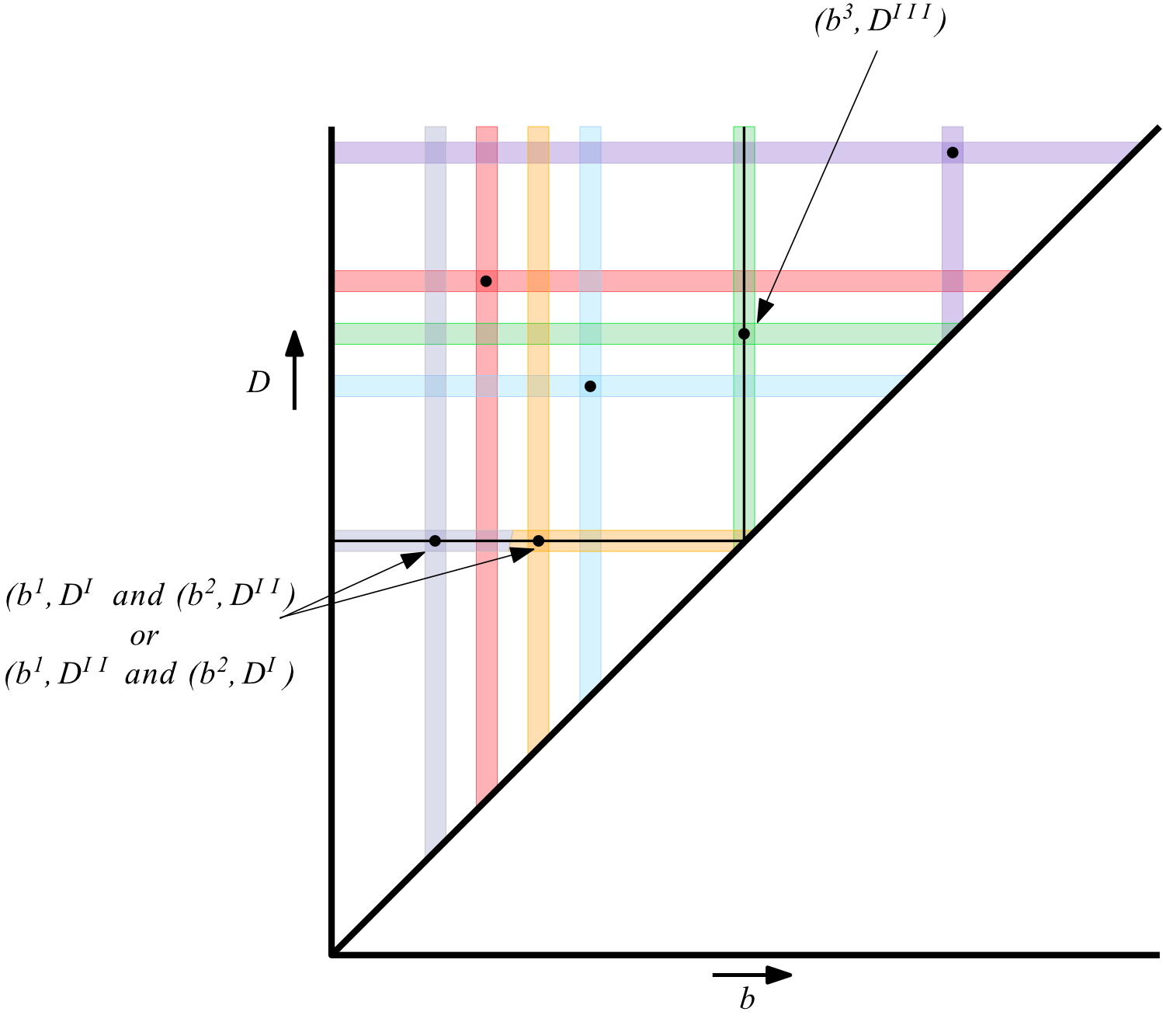}
\caption{Each point in the persistence diagram has a zone around it (either in the birth direction or the death direction and usually both) which can be assumed larger than the movement of the point in the persistence diagram. Note that we show only the above diagonal portion of the persistence diagram, for simplicity, but the portion of the diagram below the diagonal will mirror above for a manifold without boundary. 
}
\label{fig:A13PersistenceDiagram}
\end{figure}

We refer to Figure \ref{fig:A_13 Singularity}  for an illustration of the configuration and vineyard.

\begin{figure}[h!]
\centering
\includegraphics[width=0.99\linewidth]{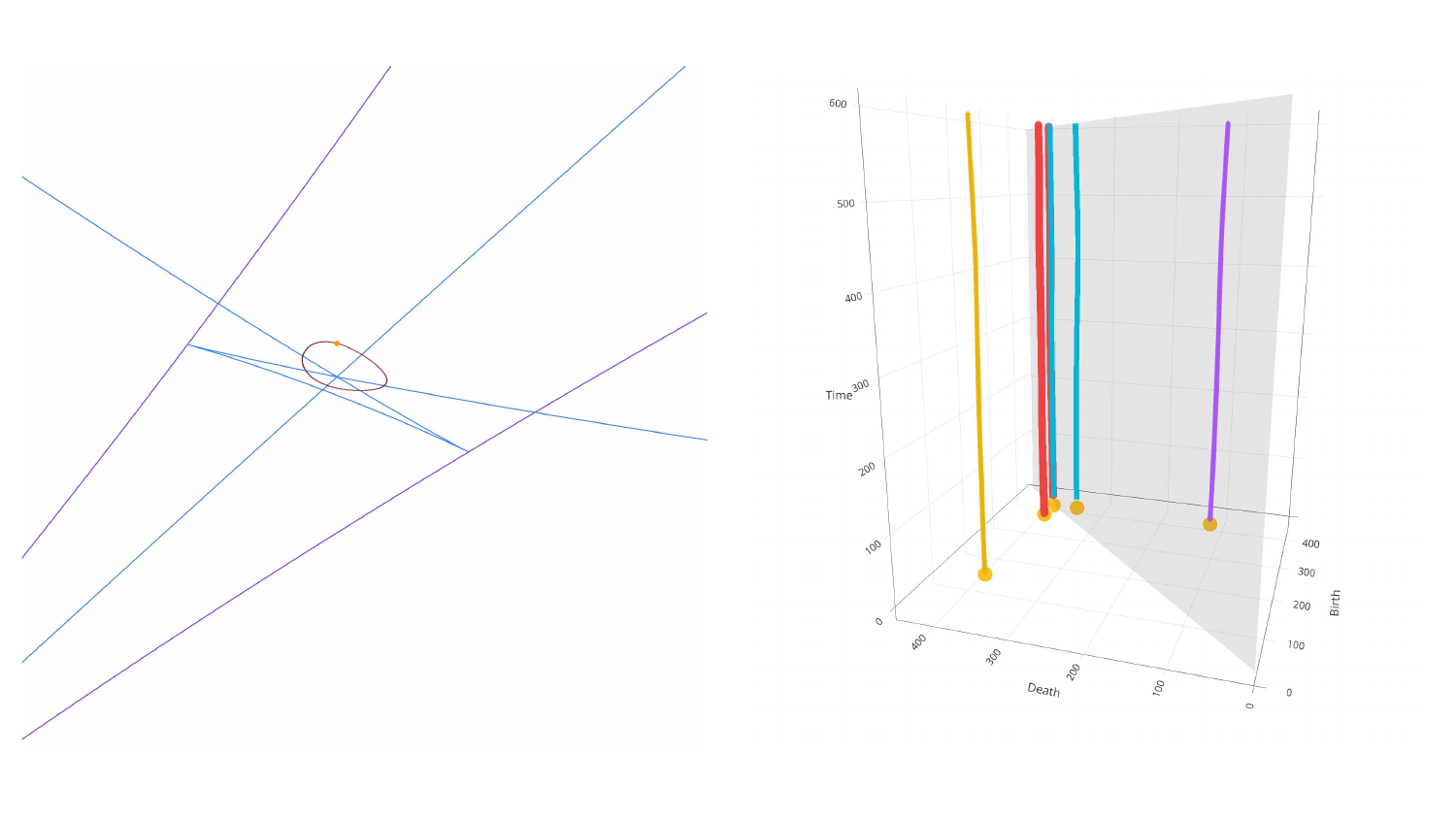}
\caption{A loop $\gamma$ enclosing an $A_1^3$ singularity (left)  and the associated vineyard (right). }
\label{fig:A_13 Singularity}
\end{figure}

\subsection{Type \texorpdfstring{$A_2$}{A2}} 
As discussed in Section \ref{sec:Singularities_Plane}, if $p$ lies on a singularity of type $A_2$, then in particular $p$ lies on a manifold piece of the focal set. As a result, for the squared distance $d_\mathbb{E} (\cdot,p)^2 |_\M$, there is a \emph{second–order} tangency at one contact point (osculating circle) and $d_\mathbb{E}$ has a \emph{degenerate} critical point there.

Fix a small neighbourhood around $p$ that contains no other critical points, and label the two sides of the focal set in that neighbourhood by $Q_\Alpha$ and $Q_\Beta$. 
Consider a small smooth loop $\gamma(t)$, around the singular point $p$, in the neighbourhood of $p$ that crosses the focal set transversally and passes alternately through $Q_\Alpha$ and $Q_\Beta$ (see Figure \ref{fig:A_2 Singularity}).
As $\gamma(t)$ crosses the focal set, the following phenomenon occurs for the vineyard: 

\begin{figure}[h]
\centering
\includegraphics[width=0.99\linewidth]{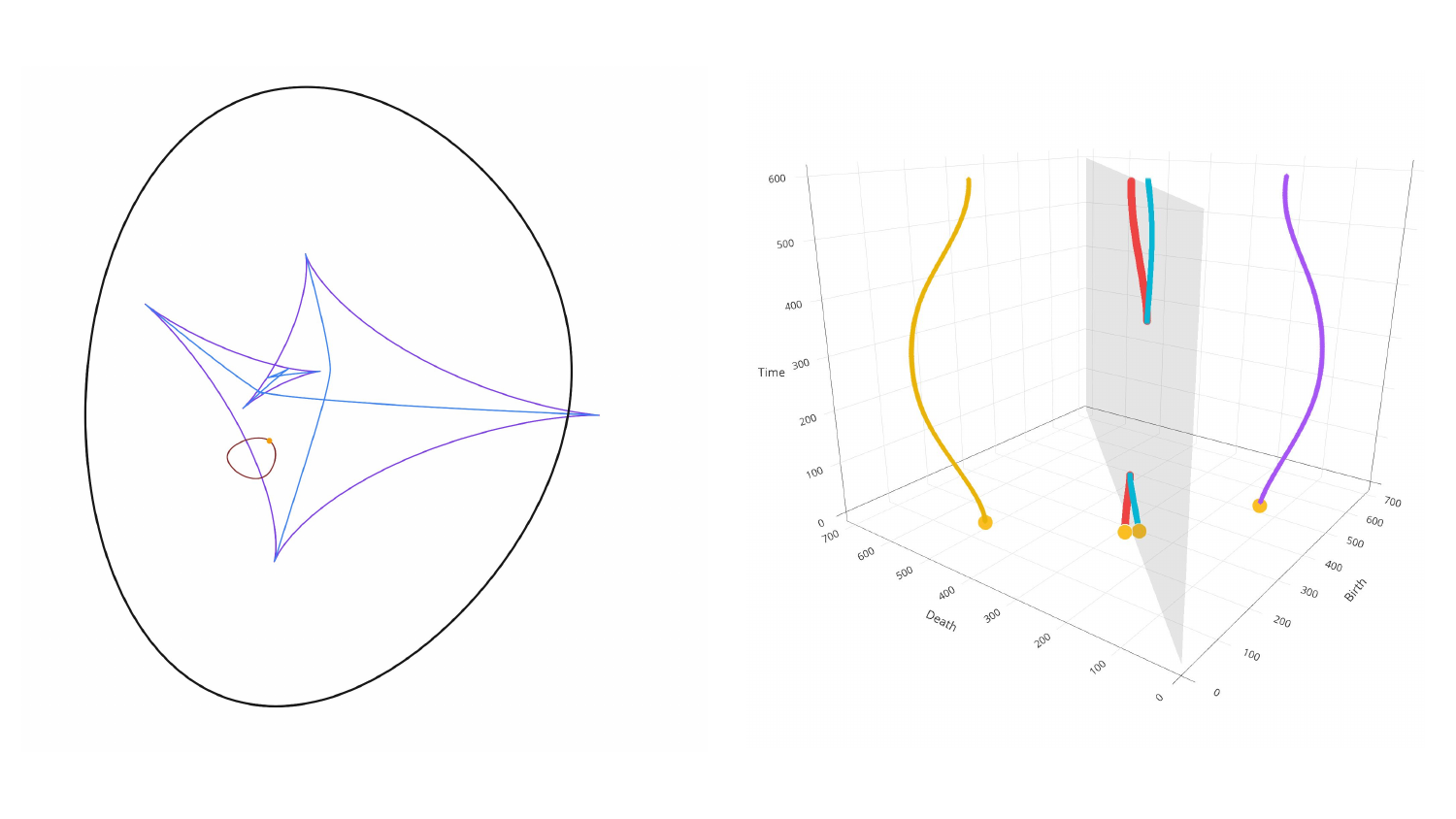}
\caption{A loop enclosing the focal set or a singularity of type $A_2$ (left)  and the associated vineyard (right). } 
\label{fig:A_2 Singularity}
\end{figure}

\begin{itemize}
\item \textbf{Persistence pair creation/annihilation.} 
When $\gamma(t)$ passes from $Q_\Alpha$ into $Q_\Beta$ across the focal set, a single new pair of simple critical points \emph{appears} locally. We denote this pair by $(b^1, D^{I})$ with $b^1<D^{I}$. This pair is \emph{self–paired}: as sublevel sets grow, a new feature appears at level $b^1$ and disappears again at level $D^{I}$ without interacting with pre–existing features, see Observation \ref{obs:Isolation}. Moreover, the difference $D^{I}-b^1$ is very small immediately after the crossing, so the corresponding point in the diagram lies very close to the diagonal. We refer to this phenomenon as the \emph{persistence pair creation}
\newline
When $\gamma(t)$ passes from $Q_\Beta$ back into $Q_\Alpha$ across the focal set, the same local pair \emph{disappears}: before the crossing the near–diagonal pair $(b^1, D^{I})$ is present, and after the crossing, it is absent. We refer to this phenomenon as the \emph{persistence pair annihilation}.
\end{itemize}

We make one additional observation. In both passages, transversality ensures that exactly one such local pair is created or annihilated and that all other critical values vary continuously and remain distinct. In particular, these crossings do not force reordering or repairing among distinct, pre–existing off–diagonal points. Because the newly created pair interacts only with itself, the elder rule does not come into effect for any pre–existing classes at these events. 

Next, along the loop $\gamma(t)$, there are exactly two transverse intersections with the focal set, one creating and one annihilating the near–diagonal pair $(b^1, D^{I})$. Hence, after one loop along the curve $\gamma(t)$, the vineyard returns to its initial state. In particular, there is no interchange between distinct off–diagonal points along such a loop, and hence no monodromy arises from circling an isolated $A_2$ singularity.
There is no linkage between the vines thanks to Observation \ref{obs:Isolation}. 


\subsection{Type \texorpdfstring{$A_1A_2$}{A1A2}}
  There is one contact point of order one and one of order two. In other words the curve (or more generally manifold) is tangent to a circle $C$ at a single point where the distance function is Morse, this corresponds to either a birth or a death (depending on the index of the Morse function). There is a contact of second order at a different point of the circle $C$, that means that the singularity is on the focal set. Moreover if we are on one side of the focal set, that is right and above in Figure \ref{fig:A_1A_2 Singularity}, then there is a single birth (because of the tangency at the circle), while if we cross the focal set there is a second birth and death which occur, that is, there is a persistence point that pops out of the diagonal. This point has very low persistence. If we cross the symmetry set, the order in which the births and deaths occur can change, however, the point with very low persistence cannot interchange with any other points in the persistence diagram, because the persistence is so low compared to all the others. Hence there is no monodromy, when going around this singularity. There is no linkage between the vines due to Observations \ref{obs:Isolation} and \ref{Obs:NoLinkingSymmetry}.

\begin{figure}[h!]
\centering
\includegraphics[width=0.8\linewidth]{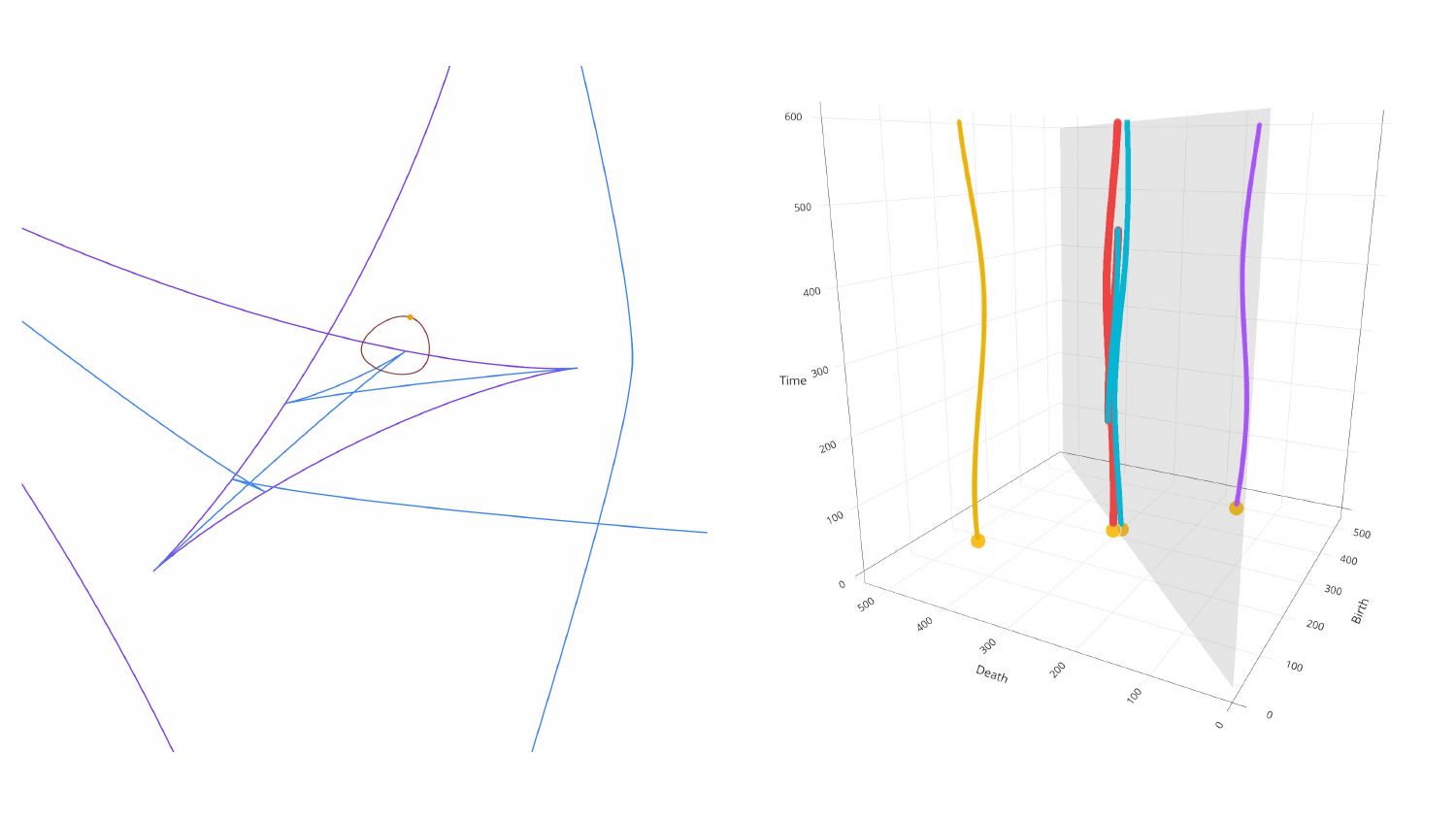}

\caption{
A loop enclosing the focal set or a singularity of type $A_1A_2$ (left) and the associated vineyard (right). 
}
\label{fig:A_1A_2 Singularity}
\end{figure}


\subsection{Type \texorpdfstring{$A_3$}{A3}} 
For this singularity, there is a maximum or minimum of curvature (where we consider the absolute value of the curvature, or non-signed curvature). There is a contact point of order 3 on the circle. A good example for the maximum of curvature is the tip of the medial axis for the ellipse (the analysis for the minimum of curvature is similar). We once more consider a $\gamma$, which encircles the point $p$ where the singularity occurs. As we go around the singularity there starting at the bottom in Figure \ref{fig:A_3 Singularity}, that is, starting at a point on the line segment connecting $p$ and the point of maximal curvature on $\M$, we first find a single closest point near the point of maximal curvature on $\M$, that is, a single birth near that point (for a minimum of curvature this is a single death in that neighbourhood). As we cross the focal set, a second birth appears as well as a death, that is, a persistence point pops out of the diagonal. As $\gamma$ crosses the symmetry set, the order of the two births interchanges, and because the point where death occurs connects the two connected components created by the births, there is an elder rule effect. As $\gamma$ crosses the focal set for the second time, the point with low persistence disappears again into the diagonal. Because the persistence of this point is so very low compared to the other points in the persistence diagram, it does not yield monodromy or topology of the vineyard. 

\begin{figure}[h!]
\centering
\includegraphics[width=0.99\linewidth]{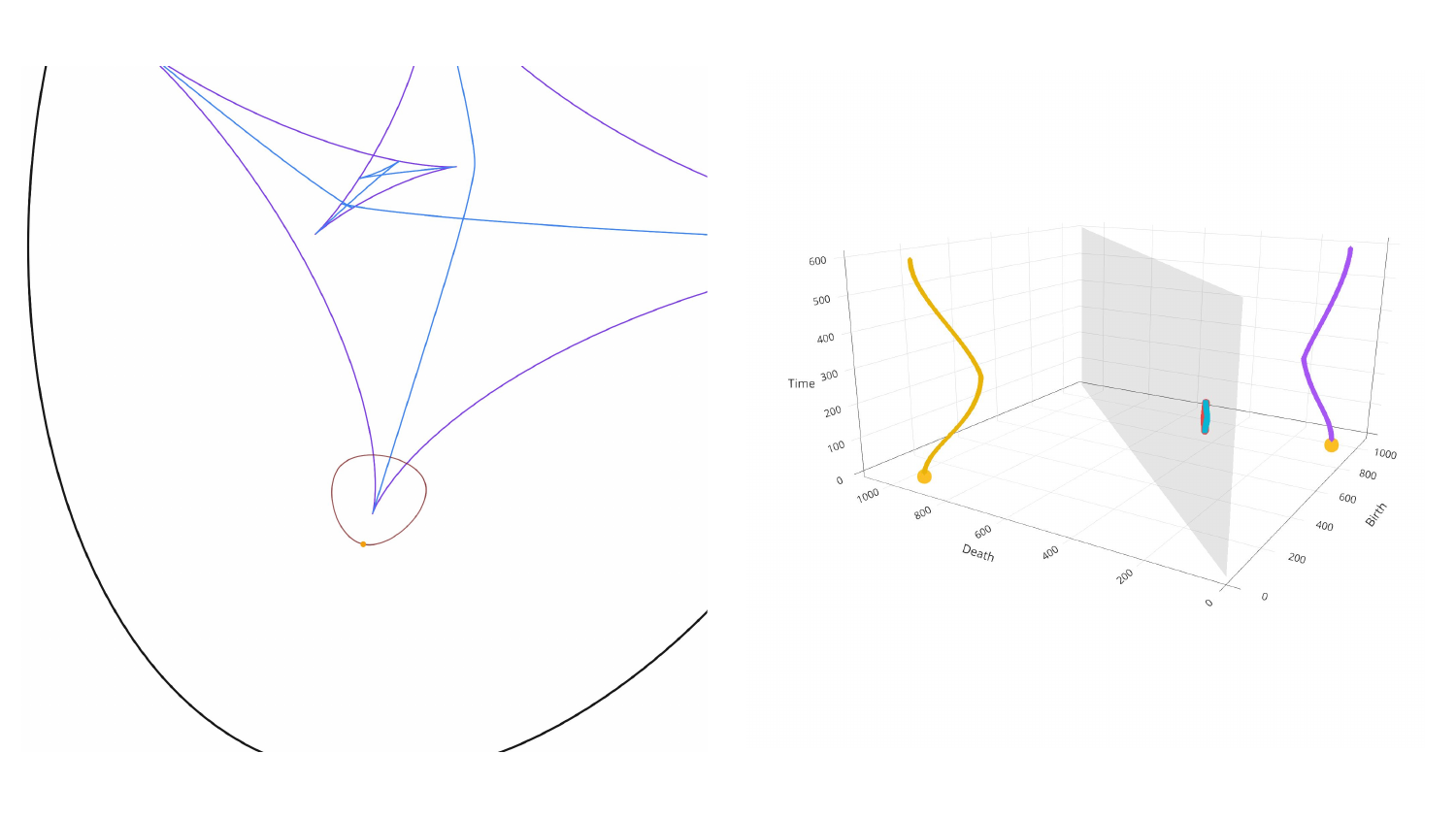}

\caption{A loop enclosing the focal set or a singularity of type $A_3$ (left) and the associated vineyard (right). 
}
\label{fig:A_3 Singularity}
\end{figure}


\subsection{Type \texorpdfstring{$A_2/A_2$}{A2A2}}
Two distinct second–order tangencies occur at the same centre $u$ but at different radii, that is there are two concentric osculating circles $C(u,r_1)$ and $C(u,r_2)$. Locally, two manifold pieces (arcs) $E_1$, $E_2$ of the focal set intersect transversely. The complement of the focal set in a neighbourhood of $u$ consists of four connected components or quadrants, denoted by $Q_\Alpha$, $Q_\Beta$, $Q_\Gamma$, and $Q_\Delta$ (oriented counter clockwise) such that $E_1$ separates $Q_\Alpha$ from $Q_\Delta$ as well as $Q_\Beta$ from $Q_\Gamma$ and $E_2$ separates $Q_\Alpha$ from $Q_\Beta$ as well as $Q_\Gamma$ from $Q_\Delta$ and $E_2$. We will assume without loss of generality that if $\eta$ is a curve crossing $E_1$ transversally from $Q_\Delta$ to $Q_\Alpha$ or from $Q_\Gamma$ to $Q_\Delta$ then there is a persistence pair creation. Similarly, we assume without loss of generality that if $\eta$ is a curve crossing $E_2$ transversely from $Q_\Alpha$ to $Q_\Beta$ or $Q_\Delta$ to $Q_\Gamma$ then there is a persistence pair creation.  
Because $r_1 \neq r_2$, we have that near $u$, the births and deaths (which are almost equal) of each persistence pair are different, even if both of them have low persistence. {See Figure \ref{fig:A_2/A_2 Singularity}}

Let us consider a small circle $\eta(t)$ that encircles the centre $u$ counter clockwise and radius $\rho$, such that $\rho$ is small compared to $r_1$, $r_2$ and $r_1,r_2$, the curvature of the two manifold pieces of the focal set and the angle between $E_1$ and $E_2$. We will follow $\eta$ starting in the delta quadrant. Our previous discussion now shows that as we move into the alpha quadrant, a point with both birth and death time close to $r_1$ emerges in the persistence diagram, or in the vineyard a vine emanates from the diagonal at birth/death $r_1$ upon crossing $E_1$. As we move into $Q_\Beta$ we have a point with both birth and death time close to $r_2$ that emerges in the persistence diagram, or in the vineyard, a vine emanates from the diagonal at birth/death $r_2$ upon crossing $E_2$. Upon moving into $Q_\Gamma$ the persistence pair near $r_1$ disappears or in the vineyard the vine that was near $(r_1,r_1)$ merges back into the diagonal. Returning to the delta quadrant, the same happens for $r_2$, that is, the vine near $(r_2,r_2)$ merges back into the diagonal. 

Because there are no other parts of the symmetry or focal set in the neighbourhood by assumption, the other points in the persistence diagram move a little, but there are no other changes in the combinatorics of the persistence diagrams. At the level of the vineyard the other vines are almost vertical and form topological circles in the closed vineyard. In particular, there is no monodromy, because the two vines appearing and disappearing from/into the diagonal do so at very distinct death/birth values, which remain almost constant owing to the small size of $\eta$. There is also no linkage between the vines due to Observation \ref{obs:Isolation}. 

\begin{figure}[h!]
\centering
\includegraphics[width=0.90\linewidth]{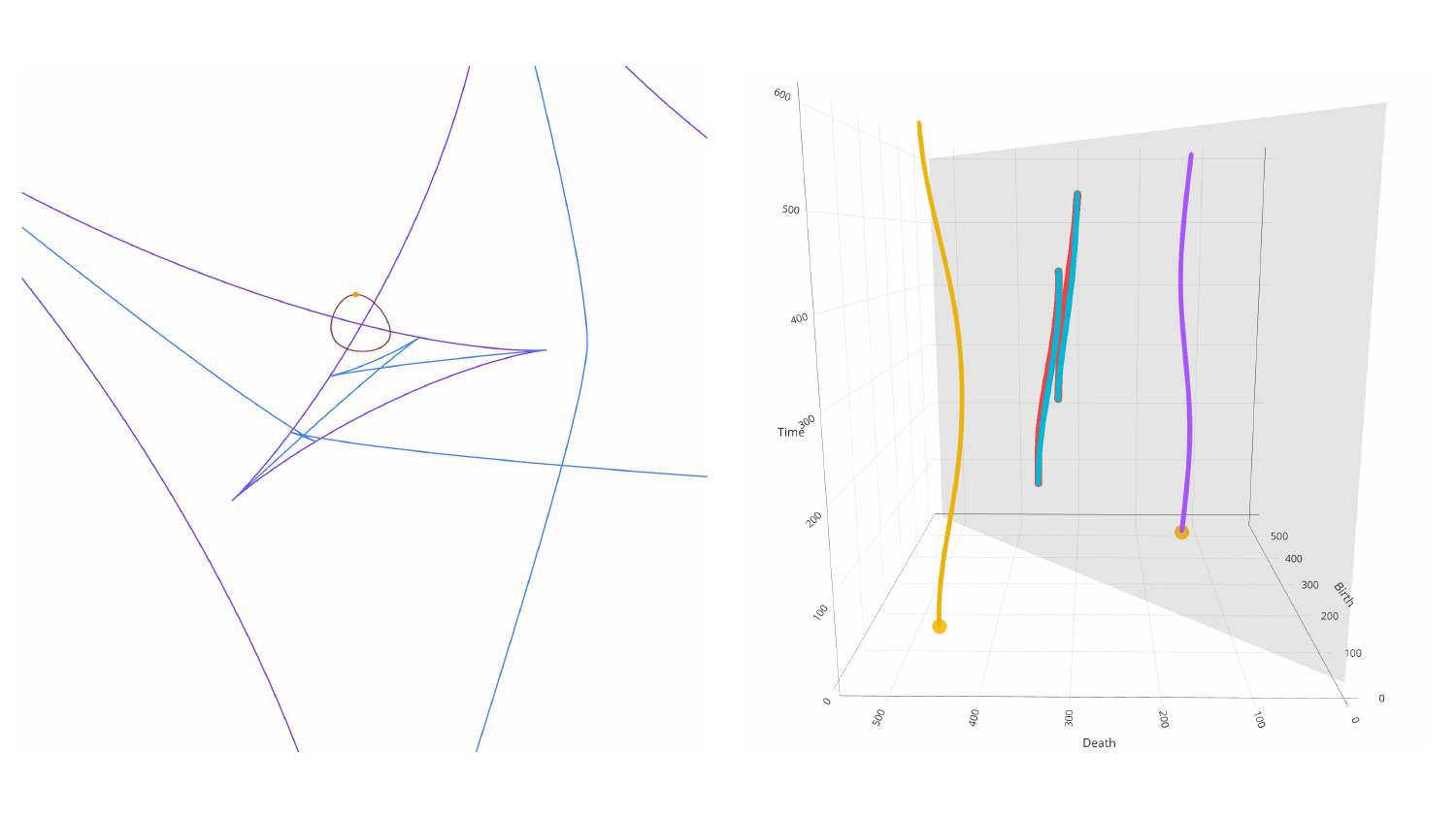}

\caption{
A loop enclosing the focal set or a singularity of type $A_2/A_2$  and the associated vineyard. 
}
\label{fig:A_2/A_2 Singularity}
\end{figure}

\subsection{Type \texorpdfstring{$A_1^2/A_2$}{A12A2}}
This is a common intersection between the focal and the symmetry set. That is there are two circles, which are concentric (but with different radii). We again consider a loop $\gamma$ enclosing the singularity. Similar to the analysis we have seen before, when crossing the focal set, a point comes out of the diagonal of the persistence diagram, but this point has very low persistence. Moreover, there can only be two births or two deaths, or one birth and one death, but not the required 2 births and 2 deaths. This means that no monodromy is generated locally at this singularity.  There is no linkage between the vines due to Observations \ref{obs:Isolation} and \ref{Obs:NoLinkingSymmetry}. 
{See Figure \ref{fig:A12/A_2 Singularity}}

\begin{figure}[h!]
\centering
\includegraphics[width=0.99\linewidth]{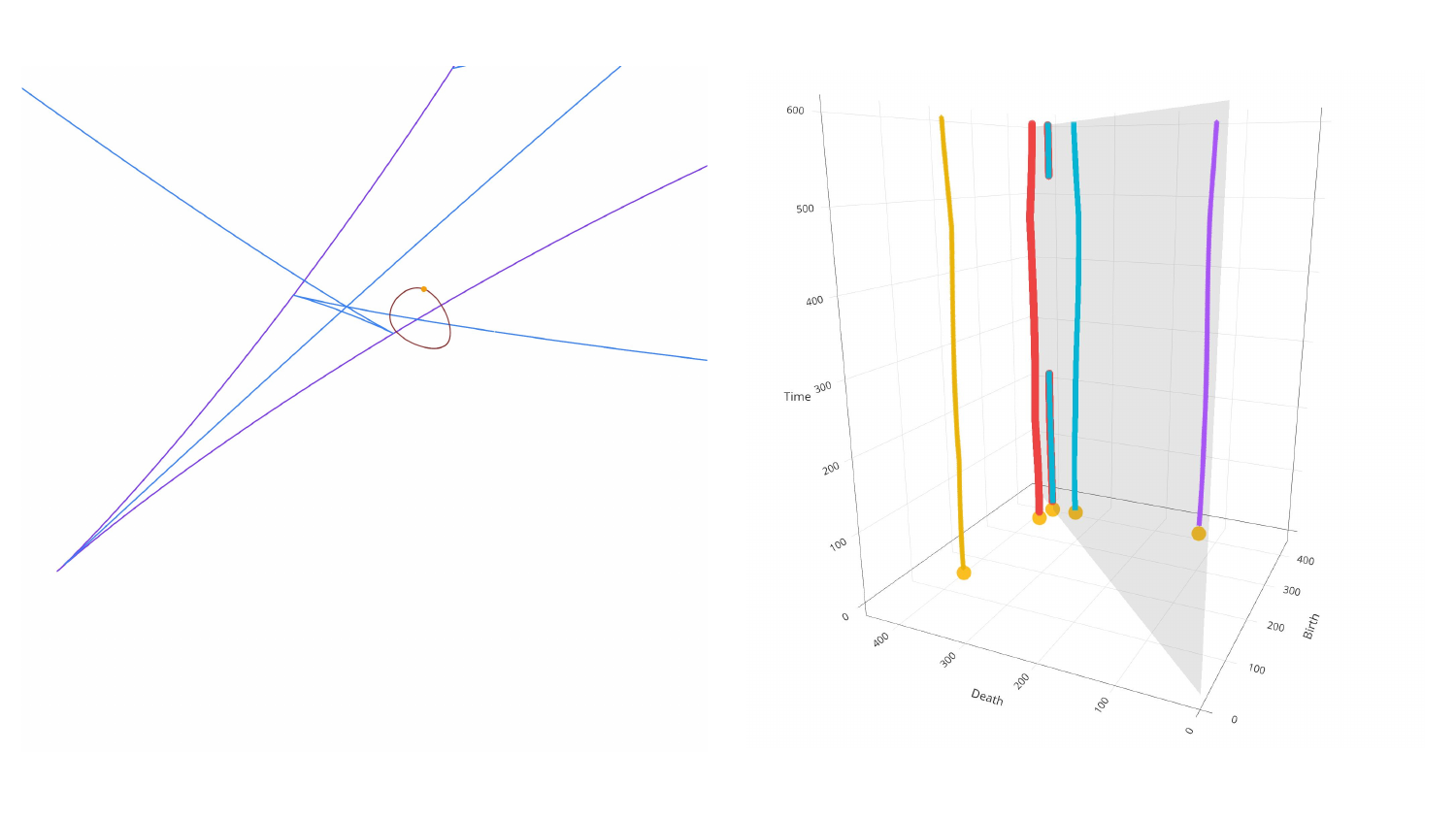}

\caption{Schematic of an $A_1^2/A_2$ singularity}
\label{fig:A12/A_2 Singularity}
\end{figure}

\subsection{Main result}
From the analysis above for all the singularities we can immediately conclude our main result:
\mainthm*

\section{Conclusions and open questions}

In this paper, we have given a local criterion which indicates the existence of monodromy for small loops enclosing a single singularity. A clear next step, which we defer to the future, is to develop more global criteria for the existence of monodromy and topology in vineyards in 2D. 

However, we stress that the converse of our main theorem is not correct. More precisely, we have the following: Consider the $1$-manifold $\M$ we considered in Section \ref{sec:A12A12}, i.e. for the $A_1^2/A_1^2$ singularity. Now choose a large loop $\gamma$ as indicated in Figure \ref{fig:a12a12_bigloop}, instead of the small loop encircling the singularity. 
Then the vineyard for this loop does not exhibit braiding or monodromy, even if there is a singularity inside that locally generates monodromy. 

\begin{figure}[h!]
    \centering
    \includegraphics[width=1\linewidth]{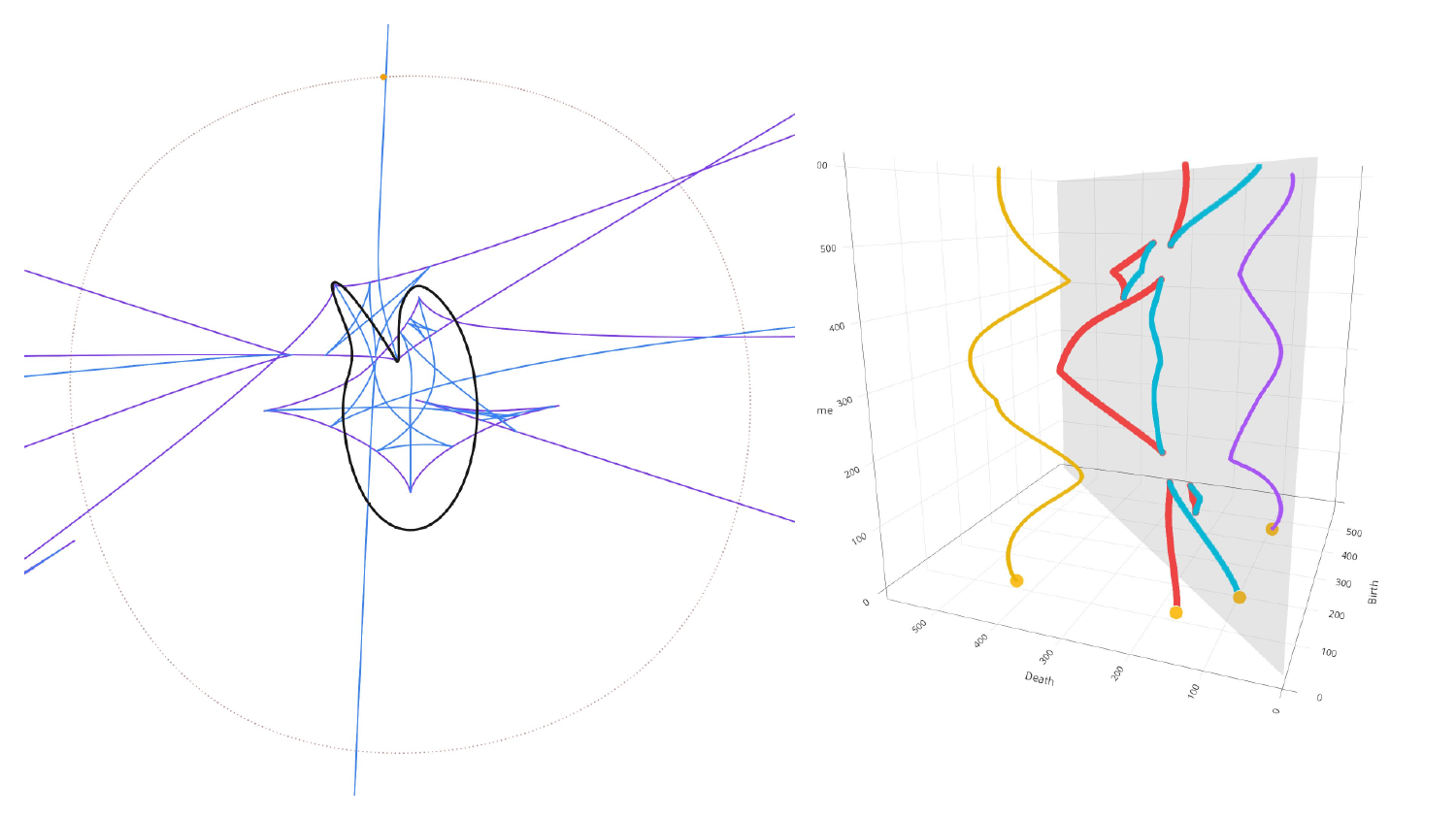}
    \caption{When $\gamma$ is a big enough loop enclosing the curve which acted as an example for a singularity of type $A_1^2/A_1^2$ which contains monodromy, then the vineyard of the larger loop no longer contains monodromy, even though it contains the $A_1^2/A_1^2$ singularity which locally generates monodromy.}
    \label{fig:a12a12_bigloop}
\end{figure}

This indicates that the way in which singularities of the generalized symmetry set interact to create the topology of the vineyard is very intricate. More precisely, a construction of the following nature does not seem straightforward to the authors: Suppose that we want monodromy of order $k$ in the vineyard. Then one would say with the information above, let us have $k$ singularities of type $A_1^2/A_1^2$ in the symmetry set. But then it is not clear how the manifold $\M$ should look so that these singularities indeed together give rise to order $k$ monodromy, let alone that we could construct all such manifolds. Specific examples were constructed by Arya et al. \cite{Arya2024} in 2D and Chambers et al. \cite{Chambers2026} in arbitrary dimensions without using singularity theoretic techniques. 

In this paper we have restricted ourselves to two dimensions, a setting which as one has seen already exhibits a very rich structure. However, generalizing these results to arbitrary dimension is a clear goal, especially given the existence of classifications of singularities in 3-dimensions by Bruce,  Giblin, and Gibson~\cite{Bruce1985symmetry}.  In fact, we envision that the higher dimensional analysis will be highly dependant upon this 2-dimensional structural result, which we plan to exhibit in a future paper.

In \cite{depthposet2026,depthposetpaper2}, persistence diagrams were augmented to depth posets, which endow the points in the persistence diagram with a partial ordering, in order to better understand how different spaces are related to each other via persistence-motivated cancellations and simplifications.  One potential area of future exploration is to adapt the depth poset to the radial transform, in order to study the depth poset for a loop around a singularity.  We conjecture that this creates a more complex form of monodromy over the poset structure, which in turn might yield more detailed geometric and topological insights.

Finally, as noted earlier, this set of obstructions to monodromy opens up the exploration of algorithmic criteria for choosing vineyards that exhibit monodromy.  While symmetry sets are notoriously noisy objects, considerable work on computing and approximating them exists due to their utility in 2D and 3D shape analysis~\cite{Giblin2002, Giblin2003,Kuijper2004,Kuijper2004a,Diatta2005}.  It remains an interesting question to determine the computational complexity of exact algorithms to compute loops containing monodromy for reasonable classes of input shapes, as well as determining its utility in broader data analysis pipelines.

\subsection*{Acknowledgements}
We thank Ellen Gasparovich and Andr{\'e} Lieutier for discussion.

\phantomsection
\addcontentsline{toc}{section}{Bibliography}
\bibliography{bib}


\appendix
\clearpage

\end{document}